%% file: main.tex
\documentclass{article}
\input{packages}

\input{macros}

\usepackage{titling}  

\title{Invitation to Local Algorithms}
\author{Václav Rozhoň}
\date{December 2023}

\begin{document}
\maketitle

\begin{abstract}
    This text provides an introduction to the field of distributed local algorithms -- an area at the intersection of theoretical computer science and discrete mathematics. We collect many recent results in the area and demonstrate how they lead to a clean theory. We also discuss many connections of local algorithms to areas such as parallel, distributed, and sublinear algorithms, or descriptive combinatorics.
\end{abstract}

\input{intro/intro}

\tableofcontents
    
\input{chapter1/1first_example}

\input{chapter1/2formal_definitions}

\input{chapter1/3sequential_vs_distributed}

\input{chapter1/4derandomization}

\input{chapter1/5network_decomposition}

\input{chapter1/6listing_problems}

\input{chapter2/1log_star_regime}

\input{chapter2/2lll_regime}
\input{chapter2/3speedups_and_slowdowns}

\input{chapter2/4classification}

\input{chapter2/5listing_classes}

\input{chapter4/1applications_to_parallel}

\input{chapter4/2applications_to_descriptive}

\input{chapter4/3other_things}

\printbibliography
\end{document}

%% file: packages.tex
\usepackage[T1]{fontenc}
\usepackage[letterpaper,margin=1.00in]{geometry}
\usepackage{amsmath, amssymb, amsthm, amsfonts}
\usepackage{bbm}
\usepackage{amscd}
\usepackage{mathrsfs}

\usepackage{comment} 
\usepackage{ifthen}
\usepackage{tikz}
\usetikzlibrary{positioning,decorations.pathreplacing}
\usepackage{ bbold }
\usepackage{appendix}
\usepackage{graphicx}
\usepackage{algorithm}
\usepackage[noend]{algpseudocode}
\usepackage{epstopdf}
\usepackage{wrapfig}
\usepackage{paralist}

\usepackage{changepage} 

\usepackage{xcolor}
\usepackage{color, colortbl}
\definecolor{Gray}{gray}{0.9}

\usepackage{mathdots} 
\usepackage{mathabx} 
\usepackage[disable]{todonotes}

\usepackage{listings} 

\usepackage{pdflscape} 
\usepackage{pifont} 

\usepackage[backend=biber,style=alphabetic,sorting=nty,maxnames=10]{biblatex}         
\let\citet\textcite
\usepackage{framed}
\usepackage[framemethod=tikz]{mdframed}
\usepackage[bottom]{footmisc}
\usepackage{enumitem}
\setitemize{noitemsep,topsep=3pt,parsep=3pt,partopsep=3pt}
\usepackage[font=small]{caption}
\usepackage{xspace}

\usepackage{thmtools} 
\usepackage{thm-restate} 

\usepackage{hyperref}
\hypersetup{
    unicode=false,          
    colorlinks=true,        
    linkcolor=red,          
    citecolor=darkgreen,        
    filecolor=magenta,      
    urlcolor=cyan           
}

\newtheorem{theorem}{Theorem}[section] 
\newtheorem{lemma}[theorem]{Lemma}
\newtheorem{meta-theorem}[theorem]{Meta-Theorem}

\newtheorem{corollary}[theorem]{Corollary}

\newtheorem{definition}[theorem]{Definition}
\newtheorem{problem}[theorem]{Problem}

\usepackage[capitalize, nameinlink,noabbrev]{cleveref}

\crefname{theorem}{Theorem}{Theorems}
\crefname{proposition}{Proposition}{Propositions}
\crefname{observation}{Observation}{Observations}
\crefname{lemma}{Lemma}{Lemmas}
\crefname{claim}{Claim}{Claims}
\crefname{problem}{Problem}{Problems}
\crefname{conjecture}{Conjecture}{Conjectures}
\crefname{question}{Question}{Questions}
\crefname{example}{Example}{Examples}
\crefname{fact}{Fact}{Facts}

%% file: macros.tex
\definecolor{darkgreen}{rgb}{0,0.5,0}

\algnewcommand\algorithmicswitch{\textbf{switch}}
\algnewcommand\algorithmiccase{\textbf{case}}

\algdef{SE}[SWITCH]{Switch}{EndSwitch}[1]{\algorithmicswitch\ #1\ \algorithmicdo}{\algorithmicend\ \algorithmicswitch}%
\algdef{SE}[CASE]{Case}{EndCase}[1]{\algorithmiccase\ #1}{\algorithmicend\ \algorithmiccase}%
\algtext*{EndSwitch}%
\algtext*{EndCase}%

\renewcommand{\tilde}{\widetilde}

\newcommand{\fA}{\mathcal{A}}

\newcommand{\fC}{\mathcal{C}}

\newcommand{\fE}{\mathcal{E}}

\newcommand{\fG}{\mathcal{G}}
\newcommand{\fH}{\mathcal{H}}

\newcommand{\fP}{\mathcal{P}}

\newcommand{\eps}{\varepsilon}

\newcommand{\e}{\textrm{e}}
\newcommand{\poly}{\operatorname{poly}}
\newcommand{\polylog}{\poly\log}

\renewcommand{\phi}{\varphi}

\newcommand{\E}{\textrm{E}}

\newcommand{\R}{\mathbb{R}}
\newcommand{\N}{\mathbb{N}}

\renewcommand{\paragraph}[1]{\vspace{0.15cm}\noindent {\bf #1}:}

\newcommand{\FullOrShort}{full}
\ifthenelse{\equal{\FullOrShort}{full}}{
  
  \newcommand{\fullOnly}[1]{#1}
  \newcommand{\shortOnly}[1]{}

  }{

    \newcommand{\fullOnly}[1]{}
    \newcommand{\IncludePictures}[1]{}
   
  }

\newcommand{\tbox}[1]{
\vspace{0.5em}
\begin{minipage}{\textwidth}
\begin{mdframed}
\begin{quote}
    \emph{#1}
 \end{quote}
\end{mdframed}
\end{minipage}
}

\newboolean{survey}
\setboolean{survey}{false}

%% file: intro/intro.tex


%% file: chapter1/1first_example.tex
\section{Local Complexity Fundamentals}
\label{chap:1_local_complexity_fundamentals}

The field of local algorithms is an area on the border of theoretical computer science and discrete mathematics where a lot of progress has happened in the past decade. 
This text is trying to serve as an introductory material presenting a certain view of the field. It aims to be helpful to beginning researchers in the area or researchers working in adjacent areas. 

There are already many resources on various aspects of local algorithms: the classical book of \citet{peleg2000book}, survey of \citet{suomela_survey}, book of \citet{barenboimelkin_book}, lecture notes by \citet{ghaffari2020DistributedGraphAlgorithms_lecturenotes}, introductory text of \citet{suomela2020round_elimination_intro}, or a recent book by \citet{hirvonensuomela2020book}. 
Unlike other texts, this one primarily explores the field's conceptual framework and complexity-theoretical aspects, rather than delving into individual problems. 
If you find errors in the text, please let me know -- this is the first version of it so there will be many! Several researchers generously gave me feedback on a preliminary version of this text; my thanks go to Anton Bernshteyn, 
Yi-Jun Chang, 
Jan Grebík, 
Yannic Maus,
Seth Pettie, Jukka Suomela, and Goran Zuzic, and especially to Mohsen Ghaffari and Seri Khoury. 

\vspace{1em}

In the first section, we introduce local algorithms and local problems in \cref{sec:1first_example}. 
We then carefully discuss the appropriate formal definitions in \cref{sec:1formal_definitions}. The following sections \cref{sec:1sequential_vs_distributed,sec:1derandomization,sec:1network_decomposition} discuss the basic theory of local algorithms and aim to convey that we are after a very clean, fundamental, and robust concept. 
Finally, \cref{sec:3concrete_problems} surveys some known results for concrete local problems.

\subsection{First Example}
\label{sec:1first_example}

Consider a very long oriented cycle that we want to properly color with as few colors as possible (see \cref{fig:intro-example}). Two colors are enough if the number of vertices $n$ is even, otherwise we need three colors.  
However, there is something uneasy about the $2$-coloring solution even when it is possible -- the solution lacks any flexibility.  
A decision to color any particular vertex with one of the two colors already implies how all the other vertices are going to be colored. 

This lack of flexibility can be undesirable for all kinds of reasons, typically when we want to design a coloring algorithm that is in some way parallel or distributed. 
If we enlarge our palette to three colors, the problem seems to go away though: Now, coloring one vertex red still implies that its neighbors are not red, but other vertices can have an arbitrary color. 

Imagine that there is a computer in every vertex of the cycle and neighboring computers can communicate. 
The computers are trying to solve the coloring problem together. How many rounds of communication are needed until each computer outputs its color? 
A message-passing algorithm of this sort is known as a \emph{local algorithm} and we define it formally in \cref{sec:1formal_definitions}. 

It is possible to convince oneself that in the case of the 2-coloring problem, at least around $n/4$ rounds are necessary to solve it, even if $n$ is divisible by two.\footnote{Consider two opposing nodes $u,v$ in the cycle graph: In less than $n/4$ rounds of communication, there is no third node that could send a message to both $u$ and $v$. Intuitively, the two vertices then cannot know whether their distance is even or odd. This argument can be made into a rigorous lower bound after the model of local algorithms is properly defined in \cref{sec:1formal_definitions}. } But what about the 3-coloring scenario? Can we solve that problem in $10$ rounds of communication? Or $O(\log n)$? Or is it similarly hard to 2-coloring? 





There indeed is a simple randomized local algorithm that solves our $3$-coloring problem after $O(\log n)$ rounds of exchanging messages. Let's describe it next. 
The algorithm should serve as an example that nontrivial local algorithms are indeed possible, though we will later see a better algorithm for the coloring problem. 

\begin{figure}
    \centering
    \includegraphics[width = \textwidth]{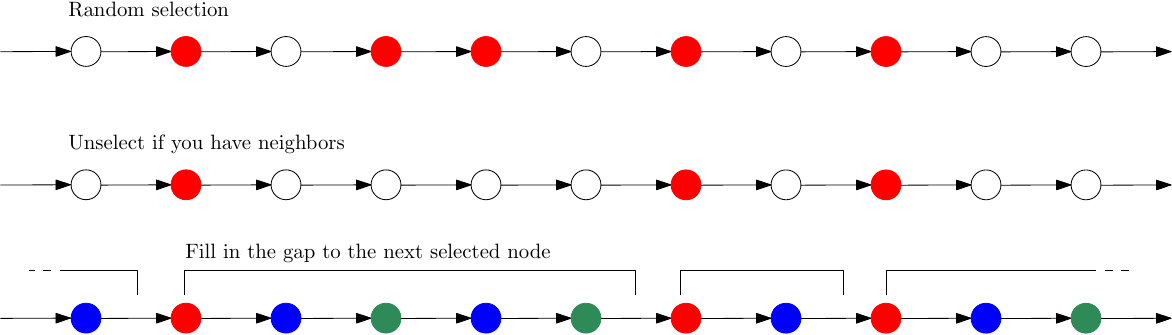}
    \caption{An example local algorithm that uses three colors to color a long cycle, a small part of which is shown. First, every vertex flips a coin and selects itself with probability $1/2$. Second, a vertex unselects itself whenever a neighbor is selected. Third, selected vertices color themselves red and each selected vertex is then responsible for coloring the subsequent vertices until the next selected one with alternating colors. }
    \label{fig:intro-example}
\end{figure}

\vspace{1em}

The $3$-coloring algorithm has two phases. In the first phase, every computer flips a coin and selects itself with probability 1/2 (top picture in \cref{fig:intro-example}). Subsequently, the vertex asks its neighbors whether they are also selected. If at least one neighbor is selected, the vertex unselects itself (middle picture in \cref{fig:intro-example}).

In the second phase of the algorithm, every selected vertex first colors itself red. 
Then, it is responsible for coloring the yet uncolored vertices until the next red vertex. 
The red node sends a message in the direction of edges, asking the subsequent vertices to color themselves by alternating the two remaining colors (see the bottom picture in \cref{fig:intro-example}). 
This algorithm properly colors the oriented cycle with $3$ colors and the number of communication rounds that it needs is asymptotically at most as large as the length of the longest run of non-red vertices in our coloring. 


To understand this quantity, consider any run of $\ell$ consecutive nodes and let us upper bound the probability that the run does not contain any red node. 
This is done by splitting the run into consecutive triples of vertices. For every triple, we know that with probability $1/8$, its middle node is initially selected, while its two neighbors are not. The selected middle node then remains selected after the end of the first phase and is colored red. 
Making this argument for every triple and using that the appropriate events are independent, we conclude that the probability of no red node in the run is at most $(7/8)^{\lfloor \ell/3\rfloor}$. 
Taking $\ell = O(\log n)$ and union bounding over all $n$ different runs of vertices of length $\ell$, we conclude that all of them contain at least one red node with $1 - 1/\poly(n)$ probability. We will call this guarantee ``with high probability'' later on. That is, our algorithm finishes after $O(\log n)$ communication rounds, with high probability.   

\vspace{1em}

Surprisingly, the fastest local algorithm for the $3$-coloring problem has a much better, albeit not constant complexity of $O(\log^* n)$.\footnote{The function $\log^* n$ measures how many times we need to take the logarithm of $n$ until we get a value of size at most $2$, i.e., $\log^* 2^2 = 1, \log^* 2^{2^2} = 2$ and so on.  }
However, instead of focusing on specific algorithms, this text is trying to give a bit more general understanding of what is going on here. 
For example, it turns out that if we come up with any local algorithm with complexity $O(\log n)$ for any reasonable problem defined on the cycle as we just did, there is a general theory that can turn this algorithm into a faster, $O(\log^* n)$-round, algorithm for the very same problem (see \cref{thm:classification_paths}). Clearly, something interesting is going on here!



%% file: chapter1/2formal_definitions.tex
\subsection{Formal Definitions}
\label{sec:1formal_definitions}

In this section, we formally define local problems and algorithms. 

\paragraph{Local problems}
We will be mostly interested in the so-called local problems. Informally speaking, these are the problems on graphs such that if the solution is incorrect, we can find out by looking at a small neighborhood of one vertex. 

Given a graph $G$ and its node $u \in V(G)$, the \emph{ball} $B_G(u, r)$\footnote{We write $B(u,r)$ when $G$ is clear from the context.  } around $u$ of radius $r$ is the subgraph of nodes around $u$ up to distance $r$. 
More generally, an \emph{$r$-hop neighborhood} is a graph with one highlighted node $v$ such that the radius of that graph measured from $v$ is at most $r$.

\begin{definition}[A local problem]
\label{def:local_problem}
    Local problem\footnote{Our definition is a simplified variant of the definition of the so-called locally checkable labeling problem by \citet*{naorstockmeyer}, discussed later in \cref{subsec:classification}. } $\Pi$ with checkability radius $r$ is formally a triplet $(S, r, \fP)$. Here, $S$ is a finite set of allowed labels and each $\fP$ is a set of allowed $S$-colored $r$-hop neighborhoods. 
    A solution to $\Pi$ in a graph $G$ is an assignment of color from $S$ to every vertex of $G$ such that for every $u \in V(G)$ we have $B_G(u, r) \in \fP$. 
\end{definition}

For example, $3$-coloring is a local problem for $S = \{\texttt{R},\texttt{G},\texttt{B}\}$, $r = 1$, and $\fP$ containing all properly colored $1$-hop neighborhoods. 
 On the other hand, a non-example of a local problem is coloring an input graph on $n$ vertices with $n$ colors: the local problem should not have different constraints for graphs of different sizes.  
Of course, while the theory of local algorithms is simplest for local problems as we defined them, the applications of local algorithms are not limited to local problems. 




\begin{figure}
    \centering
    \includegraphics[width = \textwidth]{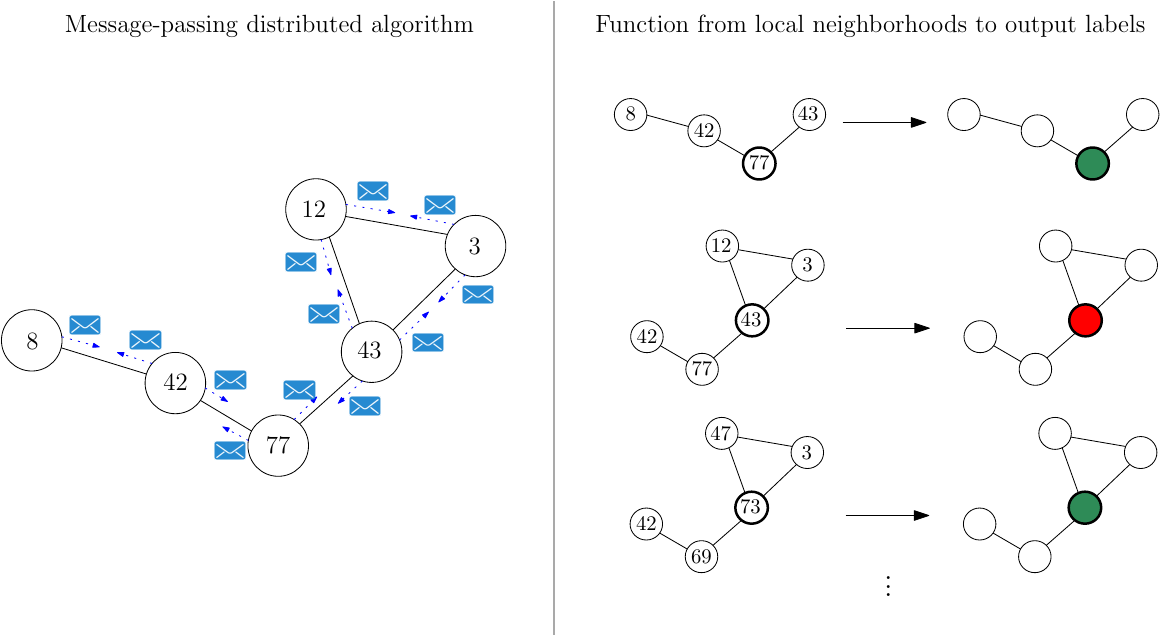}
    \caption{This picture shows the two fundamentally different ways of understanding local algorithms. \\
    Left: A $t(n)$-round local algorithm can be seen as  a distributed protocol where in each round, each node can send any message to any of its neighbors. The computers start with the knowledge of their unique identifier (or a random string). \\
    Right: A local algorithm with round complexity $t(n)$ can be seen as a function that maps each possible $t(n)$-hop neighborhood to an output label. 
    Applying this function to every vertex of the input graph always has to solve our problem: For example, if our problem is a coloring problem, the first two local neighborhoods in the above table need to map the two vertices with labels 77 and 43 to different colors, since the two labeled 2-hop neighborhoods could be a part of the same graph (which is, in fact, shown on the left). }
    \label{fig:two-definitions}
\end{figure}

\paragraph{Local algorithms}
There are two equivalent ways of thinking about local algorithms\footnote{In the literature, these algorithms are often referred to as ``distributed algorithms in the LOCAL model of computing''. We use the shorter and less formal term ``local algorithm'' for better readability.  } and both of them are important (see \cref{fig:two-definitions}). An intuitive, algorithmic definition was already sketched in \cref{sec:1first_example}: We assume that there is a computing device at every node. 
For simplicity, these devices are assumed to have unbounded computational power, thus excluding Turing machines from the definitions.
A $t(n)$-round local algorithm 
is a protocol where these devices communicate for $t(n)$ synchronous message-passing rounds using the edges of the input graph to send messages. 
At the beginning, the device in each node starts only with the information about its identifier/random string and the size of the graph, $n$. When the protocol finishes, each device outputs its part of the solution (e.g., its color). 


It will also be helpful to understand an alternative and equivalent definition that extracts the essence of what we are measuring with local algorithms. In this alternative definition, a local algorithm with round complexity $t(n)$ is simply a function that we can apply to every ball $B(u, t(n))$ of the input graph to compute the output at a given node $u$. 
Let us state it formally. 
\begin{definition}[Local algorithm]
\label{def:local_algorithm}
    A local algorithm $\fA$ with a round complexity 
    of $t(n)$ 
    is a function that accepts two inputs: firstly, the value $n$, and secondly, a labeled $t(n)$-hop neighborhood. 
\end{definition}
When we use this second definition, \emph{running} a local algorithm on an input graph $G$ simply means coloring each node $u \in V(G)$ with the output of $\fA_n(B_G(u, t(n)))$ where we used the first parameter of a local algorithm, the size of the underlying graph $n = |V(G)|$, as a subscript. \emph{Solving} a problem $\Pi$ on $G$ simply means that after running $\fA_n$ on $G$, the output colors satisfy constraints $\fP$ on all vertices of $G$. 

Moreover, in the case of \emph{deterministic} local algorithms, we assume that the nodes of the input graph are additionally labeled with unique identifiers from the range $[n^{O(1)}] = \{1, 2, \dots, n^{O(1)}\}$.\footnote{While assuming polynomial-range identifiers may look a bit arbitrary, we will see in \cref{chap:2_below_logn} that the notion of deterministic algorithms is very robust. We simply need a way of breaking the potential symmetry of the input graph -- think of what happens when you run a local algorithm on a vertex-transitive graph such as a cycle without any identifiers or randomness!} 
In the case of \emph{randomized} local algorithms, we assume that the nodes of the input graph are labeled with infinite bit strings. Solving a problem then means solving it with overall error probability at most $1/n^{O(1)}$, if the bit strings are sampled independently randomly.\footnote{Formally-minded readers may feel uneasy about the definitions not specifying the constant in the $n^{O(1)}$ expressions. We will see later in \cref{thm:slowdown} that the exact constant in the definition typically does not matter. Formally, when we say that there is a local algorithm, it means that for every $C$ there is an algorithm in the setup where we require the size of identifiers to be at most $n^C$ (or the error probability to be at most $1/n^C$). }

We notice that if there is a deterministic local algorithm solving some problem with round complexity $t(n)$, there is also a randomized local algorithm solving the same problem with the same round complexity. 
This is because any randomized algorithm can start by each node generating a random identifier from the range $[n^C]$: The probability that these identifiers are not unique, i.e., some two nodes have the same identifier, is at most $n^2 \cdot \frac{1}{n^{C}}$. Choosing $C$ large enough, this error probability can be made as small as any polynomial function of $n$. 

Finally, we remark that we can talk about local algorithms solving problems on graphs with additional structure (e.g. directed graphs) or on concrete graph classes. 
For example, in our introductory example from \cref{sec:1first_example}, it makes sense to think of all definitions relative not to the class of all graphs but to the class of graphs that are oriented paths. One interesting setup that we discuss mostly in \cref{chap:2_below_logn} is the class of bounded-degree graphs where we fix some constant $\Delta$ and analyze the class of graphs of degree at most $\Delta$. Notice that on these graphs, the set $\fP$ from the definition of local problems, as well as the support of the function $\fA$ from the definition of local algorithms, are finite. 

\paragraph{Equivalence of the two definitions}
Let's see a proof sketch of why the two definitions are equivalent. 
On the one hand, let's say we are given a function $\fA$ that maps $t(n)$-hop neighborhoods to output labels and we want to construct a $t(n)$-round message-passing protocol. Consider the protocol where in the $i$-th round, each vertex $u$ sends its neighbors everything there is to know about the ball $B(u, i)$: How the graph looks like and the values of identifiers/random strings at every node of $B(u,i)$. 
Each node $v$ can then internally use this information from its neighbors to learn everything there is to know about the ball $B(u, i+1)$.  
After $t(n)$ rounds of communication, each vertex $v$ thus knows its whole $t(n)$-hop neighborhood $B(v, t(n))$. 
At this point, the vertices stop communicating and each one applies the function $\fA$ locally to its ball which solves the problem solved by $\fA$. 

On the other hand, assume that we have a $t(n)$-round communication protocol and want to turn it into a function $\fA$ that takes $t(n)$-hop neighborhood as inputs. 
We notice that if we know the $t(n)$-hop neighborhood $B(u, t(n))$ of a node $u$, we can simulate the first round of the protocol in that ball and get to know the state of all vertices in $B(u, t(n)-1)$ after the first round. Continuing like this inductively, we conclude that starting with the knowledge of $B(u, t(n))$, we can learn the state of the protocol at the center node $u$ after $t(n)$ rounds.

\tbox{There are two different but equivalent ways of understanding local algorithms. \\
1. They are message-passing protocols running for some number of rounds. \\
2. The output at each node is a function of its local neighborhood. }

Importantly, it will be very helpful for us to keep \emph{both} definitions in mind: When we design local algorithms, the message-passing definition is more helpful as it is natural to think as ``first, we run the protocol $\fA^1$, then the protocol $\fA^2$''. On the other hand, when we prove lower bounds, the formal definition of \cref{def:local_algorithm} is easier to use. 
When we think of applications to distributed/parallel algorithms, the protocol-design definition is preferable since this is how the actual parallel/distributed algorithms are implemented. In some other applications, like applications to descriptive combinatorics, the formal definition is perhaps a bit more natural.



    

%% file: chapter1/3sequential_vs_distributed.tex
\subsection{Sequential vs Distributed Local Complexity}
\label{sec:1sequential_vs_distributed}

This section presents one of the most fundamental results in the area of local algorithms. 
Currently, it may be very unclear what kind of problems can be solved with a local algorithm of round complexity, say, $\poly\log n$.
This will become much clearer, since we will next see that, up to $\poly\log n$, the model of local algorithms is the same as the model of so-called \emph{sequential local algorithms} that are much easier to understand. 

\paragraph{The case of maximal independent set}
As an example, let's think of a concrete local problem known as the \emph{maximal independent set} problem. 
In this problem, every node must be labeled either \texttt{selected} or \texttt{unselected}.
The constraint is that each \texttt{selected} node should not neighbor any other \texttt{selected} node, while each \texttt{unselected} node should neighbor at least one \texttt{selected} node\footnote{This is a much easier problem than the \emph{maximum independent set} problem where we additionally maximize the number of \texttt{selected} nodes. }. 

Is there a local algorithm constructing a maximal independent set in a polylogarithmic number of rounds? This is not clear at all! The answer to this question is positive, and perhaps the simplest algorithm is the randomized algorithm of Luby \cite{luby86_lubys_alg,alon86lubys_algorithm}. This algorithm in fact served as the foundational example that later led \citet*{linial92LOCAL_definition} to define local algorithms. But Luby's algorithm is a non-trivial algorithm\footnote{We can briefly describe the algorithm: It runs in $O(\log n)$ rounds and in each round, every vertex chooses a random number from $[0,1]$. If its number is the largest among its neighbors, the vertex goes in the independent set and is removed from future iterations, together with its neighbors. After $O(\log n)$ rounds, all vertices are removed with high probability and the algorithm terminates.  } and just by staring at the maximal independent set problem, it is quite unclear whether a fast local algorithm exists, or not. 

This stands in stark contrast with the ``sequential'' world: If we do not care about all vertices outputting the answer ``at once'', we can compute a maximal independent set with the following simple algorithm: We choose any order of vertices and iterate over them in that order. Whenever we consider a vertex $u$, we look at its neighbors, and if at least one of them is already \texttt{selected}, we mark $u$ as \texttt{unselected}. Otherwise, we mark $u$ as \texttt{selected}. 

Here is a curious property of the above algorithm: We can still think of it as a ``local'' algorithm. 
Indeed, each node makes its decision by examining its $1$-hop neighborhood.
The only difference is that the algorithm is a \emph{sequential local algorithm} where we iterate over nodes in an arbitrary order, not a \emph{distributed local algorithm}\footnote{We sometimes use the name \emph{distributed local algorithm} to stress that we are talking about a local algorithm and not a sequential local one. } as we defined it in \cref{def:local_algorithm} where all nodes have to output their answer at once. 

A fundamental result of local complexity is the fact that these two definitions are equally powerful, up to polylogarithmic factors. Hence, in the concrete example of the maximal independent set, we can think of this problem as being ``easy'' not because of a clever algorithm like Luby's, but because of the above simple sequential algorithm. 

\tbox{
The distributed round complexity of any local problem equals its sequential local complexity, up~to~$\polylog(n)$. 
}

\paragraph{Formal definition of sequential local algorithms}
We next make this principle formal. We need to start by defining a general sequential local algorithm. 
Here is a definition of a deterministic sequential local algorithm, made slightly more powerful than the maximal-independent-set algorithm by allowing the output of each node to be not just the final color, but the node can also keep additional information that future vertices can read from that node. 

\begin{definition}[Sequential local algorithms]
\label{def:sequential_algorithm}
A sequential local algorithm of local complexity $t(n)$ is a function $\fA$ that takes two inputs, $n$ and a labeled $t(n)$-hop neighborhood. 
Its output for a neighborhood $B(u, t(n))$ around a node $u$ is a pair $(s, t)$, with $s$ being the output at $u$ and $t$ being additional information stored at $u$.
An input neighborhood to $\fA$ has some nodes labeled by these pairs. 

Running a sequential local algorithm on a graph means iterating over its vertices in some order and each time applying $\fA$ to produce the answer at the particular vertex. 
When we run $\fA$ on a node $v$, the algorithm has access to all already produced pairs $(s,t)$ at the vertices in $B(u, t(n))$ on which $\fA$ has already been run. 
    Solving a problem with a sequential local algorithm means that \emph{regardless of the order} in which we choose the vertices, this process results in a solution to the problem. 
\end{definition}
Notice that we do not require unique identifiers in the definition; we will see later in \cref{thm:bad_id_into_good_id} that they are not needed for local problems.  We can also define (oblivious) randomized algorithms where first an adversary chooses an order in which we iterate over vertices; then we sample random bits in each vertex and run our sequential local algorithm. 

We will next prove the following theorem by \citet*{ghaffari_kuhn_maus2017slocal}.
\begin{theorem}[\citet*{ghaffari_kuhn_maus2017slocal}]
    \label{thm:sequential_vs_distributed_complexity}
    Let $\fA$ be a deterministic (randomized) sequential local algorithm with local complexity $t(n)$. Then, there is a deterministic (or randomized, respectively) distributed local algorithm simulating $\fA$ with round complexity $t(n) \cdot \tilde{O}( \log^3(n))$.\footnote{We use $\tilde{O}(t(n))$ to denote the complexity $O(t(n) \cdot \log^{O(1)} t(n))$. } 
\end{theorem}
We note that it is known that there are local problems such that their sequential and distributed local complexity differ by a factor of $\Omega(\log n / \log\log n)$. \cite{gavoille2009complexity} 

\paragraph{Network decompositions}
A crucial tool that we will rely on in this section and the next one is the concept of a \emph{network decomposition}. 
Network decomposition is a clustering of the input graph into clusters of small diameter\footnote{A diameter of a graph $G$ is defined as $\max_{u,v} d_G(u,v)$ where $d_G(u,v)$ is the distance between $u$ and $v$ in $G$. }  (see \cref{fig:sequential}). 

\begin{restatable}[Network decomposition]{definition}{defdecomposition}
\label{def:network_decomposition}
    A $(c,d)$-network decomposition of a graph $G$ is a coloring of $G$ with $c$ colors. We require that vertices of each color induce a graph such that each of its connected components (that we call \emph{clusters}) has diameter at most $d$. 
\end{restatable}

We defer the discussion about the existence of network decompositions to \cref{sec:1network_decomposition}. For now, we will simply state the guarantees of the currently best deterministic network decomposition construction. 


\begin{theorem}[\cite{ghaffari2024near}]
\label{thm:best_decomposition}
There is a deterministic local algorithm that outputs a ($O(\log n)$, $O(\log n)$)-network decomposition in $\tilde{O}(\log^2 n)$ rounds.     
\end{theorem}

\paragraph{Proof of \cref{thm:sequential_vs_distributed_complexity}}
Armed with the algorithm for network decomposition, from \cref{thm:best_decomposition}, let us prove \cref{thm:sequential_vs_distributed_complexity}. 

\begin{figure}
    \centering
    \includegraphics[width = \textwidth]{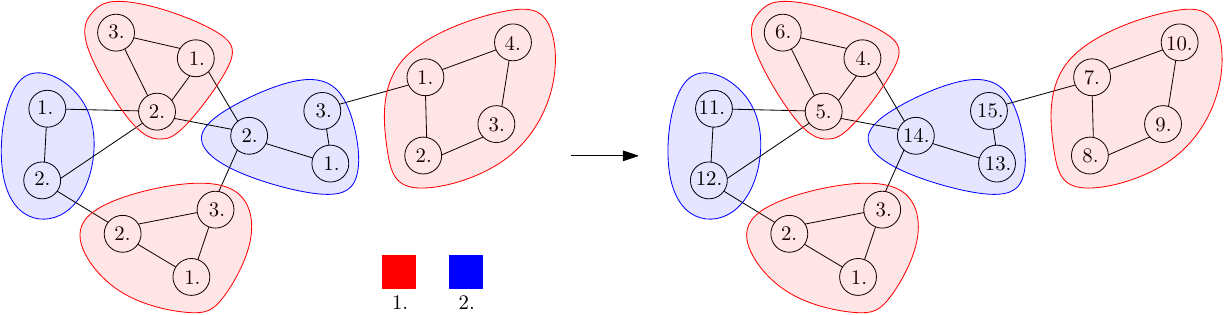}
    \caption{This picture shows a network decomposition with $c = 2$ color classes and $d = 2$ diameter. It also shows how network decomposition is used to convert an input sequential local algorithm (of local complexity $1$) into a distributed local algorithm in \cref{thm:sequential_vs_distributed_complexity}. \\
    Left: The color classes of the network decompositions are ordered as (red, blue). We iterate over the color classes and in one iteration, we consider each cluster separately and simulate an input sequential local algorithm in it (see the node ordering inside each cluster). When the algorithm is simulated in blue clusters, it has access to the output of neighboring red vertices. \\
    Right: The partial simulations of the sequential local algorithm in each cluster are consistent with a single run of that algorithm over all vertices. }
    \label{fig:sequential}
\end{figure}

\begin{proof}[Proof of \cref{thm:sequential_vs_distributed_complexity}]
    An important concept employed throughout this text is working within the power graph:
    Given a graph $G$ and a parameter $r$, we define the power graph $G^r$ to be the graph with $V(G^r) = V(G)$ where two vertices are connected whenever their distance in $G^r$ is at most $r$. 

    We start with a sequential local algorithm $\fA$ of complexity $t(n)$. We will work in the power graph $G^{t(n)}$ and construct a network decomposition in it with $c = O(\log n)$ colors and diameter $d' = \tilde{O}(\log n)$ in $G^t(n)$ via \cref{thm:best_decomposition}. 
    Consider this network decomposition in the context of the original graph $G$:
    We constructed clusters of actual diameter $d \le t(n) \cdot d' = t(n) \cdot \tilde{O}( \log n)$ in $G$. Moreover, two clusters from the same color class have distance at least $t(n)+1$ in $G$. Finally, since every communication round in $G^r$ can be simulated in $r$ communication rounds in $G$, the round complexity of constructing our network decomposition is $t(n) \cdot \tilde{O}(\log^3 n)$ by \cref{thm:best_decomposition}. 

    We will now use our clustering to simulate $\fA$. We will simulate an order of iterating over the vertices where we first iterate over all the vertices in the first color class, then all the vertices in the second color class, and so on.  
    For a fixed color class, we will arbitrarily simulate the algorithm $\fA$ in each cluster $C$ independently of all other clusters of the same color (see \cref{fig:sequential}). 

    Notice that every two clusters $C_1, C_2$ of the same color $i$ are far enough so that the $t(n)$-hop neighborhood of any vertex $u \in C_1$ never contains a vertex $u' \in C_2$. 
    Hence, simulations in different clusters of the $i$-th color do not interact and our simulation is a faithful simulation of iterating over all the vertices of $G$ in a certain order and applying the sequential algorithm $\fA$ to them.  

    Finally, let us discuss how the simulation of $\fA$ is implemented with a local algorithm. Our local algorithm simulating $\fA$ is a message-passing protocol with $c$ phases where in the $i$-th phase, each cluster $C$ of color $i$ first chooses a leader node, e.g., the node with the smallest identifier. This node collects all information about the output of $\fA$ so far up to the distance $t(n)$ from $C$. Then, the leader node internally simulates $\fA$ on $C$ and sends the result of that simulation back to all nodes in $C$. 
All this can be done in a number of rounds proportional to the diameter of $C$ and $t(n)$. Thus, the overall round complexity of the simulation is  $c \cdot t(n) \cdot \tilde{O}(\log n) = t(n) \cdot \tilde{O}(\log^2 n)$. 
\end{proof}


%% file: chapter1/4derandomization.tex
\subsection{Derandomization}
\label{sec:1derandomization}

By now, we understand that the sequential local complexity is closely related to the distributed round complexity. 
However, we still do not understand the power of randomness. 
There might be scenarios where a problem's randomized (sequential or distributed) local complexity is significantly lower than its deterministic (sequential or distributed) complexity.
Interestingly, this never happens for local problems. 
Distributed local algorithms for them can be derandomized with $\poly\log(n)$ slowdown in round complexity.\footnote{It is important to restrict ourselves to local problems. Otherwise, consider the following silly counterexample problem: We are to mark some vertices of the input graph so that at least $n/3$ but at most $2n/3$ vertices are marked. Using randomness, this problem can be solved in $0$ round complexity: every vertex simply flips a coin. But imagine trying to solve the problem deterministically: If the input graph has no edges, we are pretty screwed!
}

\tbox{
Any local problem has the same deterministic and randomized round complexity, up~to~$\polylog(n)$. 
}

Before proving this result, let us contemplate how it fits into the big picture. Thus far, we have seen \emph{six} plausible definitions of how to measure the local complexity of a problem. There are the following three different ways of thinking about it, and for each one of them, we can define both the deterministic and the randomized complexity:
\begin{enumerate}
    \item (\emph{distributed protocol}) There are computers at nodes, we design a message-passing protocol, and we measure the number of rounds of this protocol. 
    \item (\emph{distributed local complexity}) Output at each node is a function of its local neighborhood. 
    \item (\emph{sequential local complexity}) We iterate over the nodes in an arbitrary order and settle each output at a node by looking at its local neighborhood. 
\end{enumerate}

We now understand that all of these definitions are equivalent, up to $\polylog(n)$ and for local problems.  

\paragraph{Formal statement of derandomization}
Formally, we will prove the following statement by \citet*{ghaffari_harris_kuhn2018derandomizing}, using the derandomization method of conditional expectations. 

\begin{theorem}[\citet*{ghaffari_harris_kuhn2018derandomizing}]
\label{thm:derandomization_sequential}
    Let $\Pi$ be any local problem of randomized round complexity $t(n)$. Then, its deterministic sequential local complexity is $O(t(n))$. 
\end{theorem}


\begin{proof}
We are given a randomized local algorithm $\fA$ with round complexity $t(n)$ for a local problem $\Pi = (S, \fP)$ with checkability radius $r$. We will next describe a deterministic sequential local algorithm that writes an infinite sequence of bits into each node $u$ of the input graph $G$ so that if we then simulate $\fA$ with these bits, it solves $\Pi$. 

For a vertex $u \in V(G)$, define the failure indicator $X(u)$ as the indicator random variable for the event that if we run $\fA$ with random bits, it fails at $u$ at solving the local problem $\Pi = (S, \fP)$. By failure at $u$ we mean that $B(u, r) \not\in \fP$. We notice that $X(u)$ depends only on the output of $\fA$ at an $r$-hop neighborhood of $u$, and thus it ultimately depends only on random bits in the $(r + t(n))$-hop neighborhood of $u$. Moreover, the probability of failure at $u$ is less than $1/n$ by $\fA$ being a randomized algorithm solving $\Pi$. This implies that $\E\left[\sum_{u \in V(G)} X(u)\right] < 1$. 

Next, we will consider the following sequential local algorithm. We iterate over the nodes in an arbitrary order, and whenever it is a node $u_k$'s turn, we fix the random bits $B(u_k)$ at this node to a value $b(u_k)$ such that 
\begin{align*}
&\E\left[\sum_{v \in V(G)} X(v) \,\Bigm|\, \forall i \in [k]: B(u_i) = b_i \right] 
\\& \le \E\left[\sum_{v \in V(G)} X(v) \,\Bigm|\, \forall i \in [k-1]:  B(u_i) = b_i\right].
\end{align*}
In words, we set the random bits so that the expected number of errors does not increase. 

First, such a value $b(u_k)$ of random bits at $u_k$ definitely exists, since the right-hand side of the above inequality simply averages over many possible instantiations of $B(u_k)$ (that is, we rely on the law of total expectation). 
Second, we can compute this value of random bits by looking only at the $(r+t(n))$-hop neighborhood of $u_k$, since the values $X_v$ for $v$ outside of $B(u_k, r + t(n))$ are independent of the choice of random bits at $u_k$. 

After this sequential algorithm with local complexity $(r+t(n))$ finishes, we have set the random bits at every vertex $u \in V(G)$ in a way that makes $$
\E\left[\sum_{u \in V(G)} X(u) \, \Big|\, \forall i \in [n] : B(u_i) = b_i \right] < 1.
$$
But all values $X(u)$ are now deterministic, so we conclude that no failure occurs if we run $\fA$ with these bits. 

Finally, after this derandomization procedure is run, we also have to run $\fA$. We will defer the discussion of how to combine two sequential local algorithms into one to \cref{thm:sequential_pipelining}. This lemma implies that there exists a single sequential local algorithm with local complexity $O(t(n))$ that simulates first running the derandomization procedure and then running $\fA$ with the bits computed by that procedure. 
\end{proof}

Putting \cref{thm:sequential_vs_distributed_complexity,thm:derandomization} together, we get the following derandomization theorem for (distributed) local algorithms. 
\begin{theorem}[\citet*{ghaffari_harris_kuhn2018derandomizing}]
\label{thm:derandomization}
    Let $\Pi$ be any local problem of checkability $r$ and randomized round complexity $t(n)$. Then, its deterministic round complexity is $t(n)  \cdot \tilde{O}(\log^3 n )$. 
\end{theorem}
Conversely, it's known that for some local problems, the gap between randomized and deterministic local complexity can be as large as $\Omega(\log n / \log\log n)$ \cite{balliu2020separation_det_rand_by_log}.

\paragraph{Leftover: composing sequential algorithms}
We will briefly discuss how two sequential local algorithms run one after the other can be composed into a single one of larger local complexity. 

\begin{lemma}[{\citet*[Observation 2.1, Lemma 2.2]{ghaffari_kuhn_maus2017slocal}}]
    \label{thm:sequential_pipelining}
    Let $\fA_1, \fA_2$ be two deterministic (randomized) sequential local algorithms with local complexities $t_1(n), t_2(n)$. Then, there is a single deterministic (randomized, respectively) sequential local algorithm $\fA$ of complexity $2(t_1(n) + t_2(n))$\footnote{In general, $k$ local sequential algorithms can be simulated with complexity $2\sum_{i = 1}^k t_i(n)$. } that simulates the output of first running $\fA_1$, and then running $\fA_2$ on the output of $\fA_1$. 
\end{lemma}
\begin{proof}
    We would like $\fA$ to work as follows: We iterate over the nodes and when it is the turn of a node $u$, we first simulate $\fA_1$ for all nodes in $B(u, t_2(n))$. Then, we use the computed information to simulate $\fA_2$ at $u$ to compute the final output at $u$. 

The only difficulty is that once $u$ simulates $\fA_1$ for a node $v \in B(u, t_2(n))$, we cannot simulate $\fA_1$ at $v$ again in the future since we want to have the guarantee that all simulations of $\fA_1$ taken together correspond to a consistent iteration over the nodes of the input graph and running $\fA_1$ on them. 

Thus, our algorithm $\fA$ will additionally store at $u$ the output of simulations of $\fA_1$ within $B(u, t_2(n))$.
When it is the turn of a vertex $u$, 
$\fA$ starts by looking at its $2t_2(n)$-hop neighborhood and fixing the answers of $\fA_1$ for nodes $v \in B(u, t_2(n))$ at which $\fA_1$ was already simulated in the past. 
Only then we simulate $\fA_1$ for the rest of the nodes in $B(u, t_2(n))$ and run $\fA_2$ after that. 
The final local complexity of $\fA$ is $\max\left(t_1(n) + t_2(n), 2t_2(n)\right)$.  
\end{proof}


%% file: chapter1/5network_decomposition.tex
\subsection{Network Decompositions}
\label{sec:1network_decomposition}

Let us recall the definition of a network decomposition: 

\defdecomposition*

We will next discuss constructions of network decompositions\footnote{In the literature, one can very often encounter numerous variants of network decompositions with names like low-diameter clusterings, padded decompositions, or sparse neighborhood covers. }: the missing piece in proof of \cref{thm:sequential_vs_distributed_complexity,thm:derandomization}. 

\paragraph{Existence of network decompositions}
First of all, it is unclear whether network decomposition of the input graph, say with parameters $c,d = O(\log n)$ always exists. Let's confirm this by constructing it using the folklore sequential \emph{ball-carving} algorithm. 

\begin{theorem}[Ball-carving algorithm]
\label{thm:network_decomposition_sequential}
    An $(O(\log n),$ $ O(\log n))$-network decomposition exists for any graph $G$. 
\end{theorem}
\begin{proof}
    We will show how to construct the first color class of the network decomposition. In particular, we will find a family of vertex-sets $\fC = \{C_1, C_2, \dots, C_t\}$, where we call each $C_i \subseteq V(G)$ a \emph{cluster}, such that:
    \begin{enumerate}
        \item The clusters are not adjacent; i.e., there is no edge $uv \in E(G)$ with $u \in C_i, v\in C_j$, $i \not= j$, \label{item:sb1}
        \item each cluster $C_i$ has diameter $O(\log n)$, \label{item:sb2}
        \item at least $n/2$ vertices are clustered, i.e., $\left| \bigcup_{i = 1}^t C_i \right| \ge n/2$.   \label{item:sb3}
    \end{enumerate}

    To construct the clustering $\fC$, we iterate over the vertices of $G$ in an arbitrary order. 
    Each time we are at a vertex $u$, we consider the balls $B(u, 0), B(u, 1), \dots $ of increasing radii around it. In particular, think of gradually growing larger and larger balls around $u$, and once the size of the ball does not at least double, i.e., once we have $|B(u, i+1)| \le 2 \cdot |B(u, i)|$ for the first time, we let $B(u, i)$ be the next cluster of $\fC$. 
    Additionally, we remove all vertices from the boundary $B(u, i+1) \setminus B(u,i)$ from $G$. 
    After the ball around $u$ is grown, we continue this procedure with an arbitrary next vertex of $G$ that was not explored yet. \todo{tady obrázek}

    To analyze this algorithm, first notice that no two clusters are adjacent (\cref{item:sb1}), since after growing each cluster, we delete its boundary. 

    Moreover, we claim that all balls have diameter $O(\log n)$ (\cref{item:sb2}). To see this, note that the volume of each ball $B(u, i)$ grows as $|B(u,i)| \ge 2^i$ until we finish growing. 
    This implies that if we are not finished in iteration $1 + \log_2n$, the ball $B(u, 1 + \log_2 n)$ has to contain more than $n$ vertices, a contradiction. 
    
    Finally, we cluster at least half of the vertices (\cref{item:sb3}), because by definition, whenever a ball stops growing, we have $|B(u, i+1) \setminus B(u,i)| \le |B(u,i)|$, so the number of vertices unclustered because of the ball around $u$ can be charged to the number of vertices clustered around $u$.     

    We construct the desired network decomposition by simply repeating the above process $\log_2 n$ times: The clustered vertices in the $i$-th step are removed from the graph and form the $i$-th color class of the final decomposition. After $\log_2 n$ rounds, every vertex gets clustered. 

\end{proof}


\paragraph{Deterministic distributed algorithms}
There is a long line of work on deterministic algorithms for constructing network decompositions \cite{awerbuch89,panconesi-srinivasan,rozhon_ghaffari2019decomposition,ghaffari_grunau_rozhon2020improved_network_decomposition,chang_ghaffari2021strong_diameter,elkin_haeupler_rozhon_grunau2022Clusterings_LSST,rozhon_haeupler_grunau2023simple_decomposition,ghaffari_grunau_haeupler_ilchi_rozhon2023improved_decomposition,ghaffari2024near}. 
The currently fastest deterministic algorithm for network decomposition is the algorithm by \citet*{ghaffari2024near}. We stated its guarantees in \cref{thm:best_decomposition}. 

Perhaps the simplest $\polylog n$-round algorithm is the algorithm of \citet*{rozhon_ghaffari2019decomposition} and its variant by \citet*{rozhon_haeupler_grunau2023simple_decomposition}. They both need $O(\log^7 n)$ rounds and construct clusters with diameter $O(\log^3 n)$. 
We will next explain the construction from \cite{rozhon_ghaffari2019decomposition} that in fact outputs clusters that only have a weaker guarantee of small \emph{weak-diameter}. 
A weak-diameter of a cluster $C$ is defined as $\max_{u,v \in C} d_G(u,v)$, as opposed to the (strong-) diameter which is defined as $\max_{u,v \in C} d_{G[C]}(u,v)$. That is, a cluster with a small weak diameter can be even disconnected, we only require that its vertices are close in the underlying graph $G$. 
A small weak-diameter is sufficient for our applications such as the proof of \cref{thm:sequential_vs_distributed_complexity} and we will discuss at the end of this section how one can convert the weak-diameter guarantee to the strong-diameter one in a black-box way. 

\begin{theorem}[Distributed ball-carving algorithm]
    \label{thm:network_decomposition_deterministic}
    There is a local algorithm with round complexity $O(\log^7 n)$ that constructs a network decomposition with $c = O(\log n)$ colors and $d = O(\log^3 n)$ weak-diameter. 
\end{theorem}
\begin{proof}
    Similarly to \cref{thm:network_decomposition_sequential}, we show how to construct a family of clusters $\fC = \{C_1, C_2, \dots, C_t\}$ such that:
    \begin{enumerate}
        \item No two clusters are adjacent, \label{item:db1}
        \item each cluster $C_i$ has weak-diameter $O(\log^3 n)$, \label{item:db2}
        \item at least $n/2$ vertices are clustered. \label{item:db3}
    \end{enumerate}

    Recall that in deterministic local algorithms, every vertex starts with a unique $b = O(\log n)$-bit identifier. Our algorithm has $b$ phases. At the beginning of the algorithm, we start with a clustering $\fC_0$ that contains each node of $G$ as a trivial one-vertex cluster. Each cluster $C$ in our clustering is assigned a unique identifier $id(C)$ -- in particular, the identifier of a cluster is set to be the identifier of its unique vertex. 
    
    The clustering $\fC_i$ evolves during the following $b$ phases, with some vertices changing which cluster they belong to and some vertices being deleted from $\fC_i$. We will prove that after each phase $i \in [b]$, the clustering $\fC_i$ has the following guarantees:
    \begin{enumerate}
        \item Consider the graph $G_i = G[\bigcup_{C \in \fC_i} C]$, i.e., the graph induced by vertices that were not deleted yet. Consider any connected component $K$ of $G_i$. All clusters present in $K$ have the same first $i$ bits in their unique identifier. \label{item:pb1}
        \item Each cluster $C \in \fC_i$ has weak-diameter $i \cdot O(\log^2 n)$. \label{item:pb2}
        \item $\left| V(G_i) \right| \ge n - i \cdot \frac{n}{2b}$. \label{item:pb3}
    \end{enumerate}    
    Plugging $i = b$, the inductive guarantees of \cref{item:pb1,item:pb2,item:pb3} reduce to the final desired guarantees of  \cref{item:db1,item:db2,item:db3} from the beginning of the proof. Thus, we only need to show how one phase of the algorithm makes sure that if we start with the guarantees for some $i$, they also hold for $i+1$. 
    \begin{figure}
        \centering
        \includegraphics[width = \textwidth]{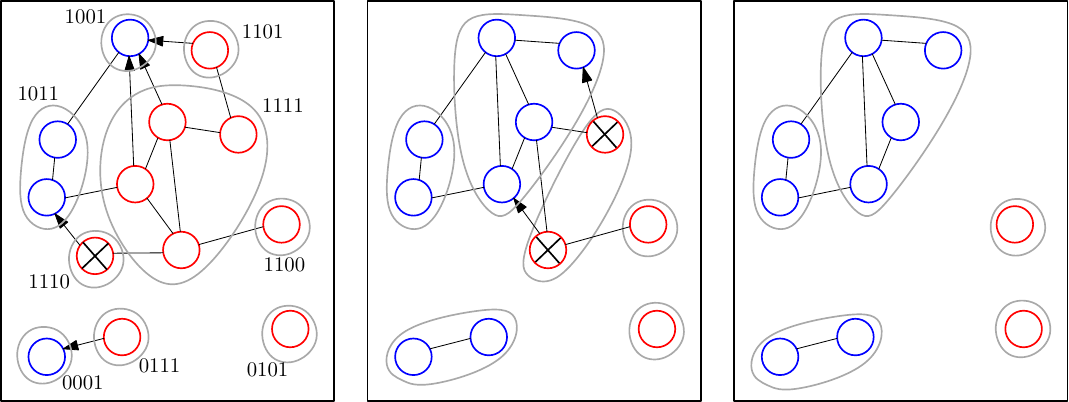}
        \caption{This picture illustrates the second phase of the algorithm from \cref{thm:network_decomposition_deterministic}, in a simple example graph. 
        Left: At the beginning of the second phase, the clusters are already separated according to their first bit. 
        In particular, if the identifier starts with $1$, the cluster is in the top connected component, while if the identifier starts with $0$, the corresponding cluster is in one of the two bottom connected components. In the second phase, a cluster is active (blue) if its second bit in the identifier is $0$, and inactive (red) otherwise. In the first step of this phase, each red vertex proposes to join an arbitrary neighboring blue cluster, provided that there is one (the arrows in the picture). The active cluster then either decides to grow (cluster $1001$) or the vertices that proposed to it are deleted (cluster $1011$). Note that if a cluster decides to delete its boundary, it does not neighbor any red nodes from that point on.  \\
        Middle: After one step of our algorithm, the inactive cluster with identifier $1111$ became disconnected but its weak-diameter remains the same as at the beginning of the phase. In the second step of this phase, each red vertex again proposes to join an arbitrary neighboring blue cluster. \\
        Right: The picture shows the resulting clustering after the second phase is finished.  Note that the only two adjacent clusters with identifiers $1011$ and $1001$ will be separated in the following, third, phase, as their identifiers differ on the third position. }
        \label{fig:distributed_decomposition}
    \end{figure}

    At the beginning of each phase $i+1 \in [b]$, we classify each cluster in $\fC_i$ as either \emph{active} or \emph{inactive}, where a cluster $C$ is active if and only if the $(i+1)$-th bit in $id(C)$ is equal to zero. 

    Next, we run a variant of the ball-carving algorithm from \cref{thm:derandomization_sequential} with $t = 10b\log_2 n$ steps. 
    In each step of this algorithm, each vertex $u$ from some inactive cluster first checks whether it neighbors with a vertex that is currently present in an active cluster. 
    If there are more such neighboring vertices, $u$ chooses an arbitrary such vertex in $C$ and then $u$ \emph{proposes to join $C$}. 

    Next, each cluster $C$ collects how many vertices proposed to join $C$. If there are at least $|C| / (2b)$ such vertices, then $C$ decides to grow: All the vertices that proposed to join $C$ leave their original cluster and join $C$. 
    On the other hand, if less than $|C| / (2b)$ proposed to join $C$, all these vertices are \emph{deleted}; they will not be present in $\fC_{i+1}$ and they will not participate in the rest of the algorithm. 
    Note that after an active cluster $C$ decides to delete the proposing vertices, it is neighboring only with nodes from other active clusters and thus it does not grow anymore in the rest of the phase. This finishes the description of one phase of the algorithm (see \cref{fig:distributed_decomposition} for an illustration). 

    To analyze the phase, we first check that each cluster stops growing during the phase. This is true because otherwise, at the end of the phase, the cluster would have to contain more than 
    \[
    \left( 1 + 1/(2b) \right)^{10b \log_2 n} > n
    \]
    nodes. In particular, the weak-diameter of each active cluster grows additively by at most $2 \cdot 10b \log_2 n = O(\log^2 n)$, while the weak-diameter of each inactive cluster does not increase. This proves \cref{item:db2}. 

    Next, consider any connected component $K$ of $G_i$. The clusters in $K$ already agree on the first $i$ bits of their identifier at the beginning of the phase $i+1$ by our inductive assumption. 
    Recall that all active clusters stop growing during the phase and delete the inactive nodes on their boundary; in particular no active and inactive cluster neighbor at the end of the algorithm. 
    Therefore, in $G_{i+1}$ the component $K$ further splits into connected components $K_1, K_2, \dots $ such that for each component $K_i$ we have that either all of its clusters are active or all are inactive. We conclude that clusters in each connected component of $G_{i+1}$ agree on the first $i+1$ bits in their identifier as needed in \cref{item:db1}. 

    Finally, whenever a cluster $C$ stops growing, we delete at most $|C| / (2b)$ nodes. Thus, during the phase, we delete at most $n/(2b)$ nodes which proves \cref{item:db3}. 

    Let us discuss the round complexity of the algorithm. The clusters have weak-diameter $O(\log^3 n)$. Hence, each growing step can be implemented with that round complexity. There are $t = O(\log^2 n)$ steps in one phase, $b = O(\log n)$ phases, and we need to repeat the overall algorithm $c = \log_2 n$ times. We conclude that the overall round complexity is $O(\log^7 n)$.  
\end{proof}

Next, let us show how we can use our current knowledge of local algorithms to construct a local algorithm for network decomposition with $c,d = O(\log n)$. 
\begin{corollary}
    \label{cor:network_decomposition_deterministic}
    There is a local algorithm with $O(\log^9 n)$ round complexity that constructs a network decomposition with $c = O(\log n)$ colors and $d = O(\log n)$ diameter. 
\end{corollary}
\begin{proof}
    We observe that the algorithm from \cref{thm:network_decomposition_sequential} can be seen as a sequence of $O(\log n)$ sequential local algorithms per \cref{def:sequential_algorithm}, each one with local complexity $O(\log n)$. Using \cref{thm:sequential_pipelining}, we can thus view it as a single sequential local algorithm of local complexity $O(\log^2 n)$. 

    Thus, we can use \cref{thm:sequential_vs_distributed_complexity} to convert this algorithm into a local algorithm. While we phrased the guarantees of \cref{thm:sequential_vs_distributed_complexity} using the best available network decomposition algorithm, we can plug in the network decomposition from \cref{thm:network_decomposition_deterministic} instead; the reduction holds even if we use a network decomposition with a weak-diameter guarantee. 
\end{proof}

Let us briefly go back to the observation from the proof of \cref{cor:network_decomposition_deterministic} that the algorithm from \cref{thm:network_decomposition_sequential} can be viewed as a sequential local algorithm. This observation helps us to appreciate network decomposition as the ``complete'' problem for turning sequential local algorithms into distributed ones: On the one hand, network decomposition allows us to do this task via \cref{thm:sequential_vs_distributed_complexity}; on the other hand, \emph{any} way of proving that theorem leads to a distributed algorithm for network decomposition via the proof of \cref{cor:network_decomposition_deterministic}.  


\paragraph{Randomized distributed algorithms}
There is a classic randomized algorithm by \citet*{linial_saks93low_diameter_decompositions} that constructs an $(O(\log n), O(\log n))$-network decomposition in $O(\log^2 n)$ distributed rounds. Let us sketch its beautiful variant by \citet*{miller2013parallel} also known as the MPX algorithm:

In their algorithm, every vertex independently chooses a random \emph{head start} sampled from an exponential distribution; that is, a head start of each node is $0$ with probability $1/2$, $1$ with probability $1/4$, and so on. 
Next, we run the breadth-first-search algorithm from all nodes at once with those head starts. That is, we first choose some number $T = O(\log n)$ such that with high probability, no node samples a head start larger than $T$. Then, we simulate a run of breadth-first search from an additional virtual node $u_0$ which is connected to every other node $u$ with an oriented edge of length $T- \text{head start}(u)$. All vertices reached by the breadth-first search from the same starting node $u$ form one cluster. The output of the algorithm are only those vertices such that all their neighbors are from the same cluster. 

One can prove that the output clustering from this algorithm contains at least a constant fraction of all vertices and uses disjoint clusters of diameter $O(\log n)$, with high probability. We repeat this process $c = O(\log n)$ times to construct a network decomposition. 

%% file: chapter1/6listing_problems.tex
\subsection{Bounds for Concrete Problems}
\label{sec:3concrete_problems}

This survey emphasizes the general properties of local algorithms, but it is crucial to remember that their study is rooted in understanding specific local problems. One of the most significant problems that has driven much of the development in this field is the maximal independent set problem. This problem belongs to the group of four classical \emph{symmetry-breaking} problems, along with maximal matching, vertex coloring, and edge coloring that we discuss later in this section. 

\subsubsection{Maximal Independent Set}
Recall that the maximal independent set problem involves finding a subset of vertices in a graph that is both independent (no two vertices are adjacent) and maximal (adding any other vertex would violate independence). 
The first algorithms for this problem were the randomized algorithms by \citet{luby86_lubys_alg} and \citet{alon86lubys_algorithm}, which find a solution in $O(\log n)$ rounds with high probability.\footnote{\cref{sec:1sequential_vs_distributed} contains additional discussion related to maximal independent set and Luby's algorithm.}
Around the same time, Linial introduced the model of local algorithms \cite{linial92LOCAL_definition}, perhaps influenced by those algorithms. 

One of the most important open questions in this area is whether a sublogarithmic-round randomized algorithm exists for the maximal independent set problem.

\begin{problem}
    Is there a randomized $o(\log n)$-round algorithm for the maximal independent set problem? Is there an $\tilde{O}(\sqrt{\log n})$-round algorithm?\footnote{This text lists some open problems in the area. For a different list, see \cite{suomela_open}.}
\end{problem}\todo{add cit from seri}

Note that this problem specifically asks for a randomized algorithm. We discuss the deterministic case shortly. 

The second part of the question is motivated by the celebrated lower bound by \citet{kuhn16_jacm}, which shows that the complexity of the problem is $\Omega\left(\sqrt{\frac{\log n}{\log\log n}}\right)$.
More precisely, they show that for any maximum degree $\Delta$, we can derive a lower bound of $\Omega\left(\min\left( 
\sqrt{\frac{\log n}{\log\log n}},
\frac{\log\Delta}{\log\log \Delta}
\right)\right)$ for the maximal independent set problem. This motivates analyzing the complexity of the problem both as a function of $n$ and $\Delta$.

\paragraph{Algorithms for small $\Delta$}
An important related question is to understand the complexity of the problem as a function of the maximum degree, $\Delta$. In this context, remarkable progress has been made using the so-called shattering framework. The main idea is to develop fast randomized algorithms that solve the problem for most vertices, leaving only small ``islands'' of unsolved vertices to be resolved later by an appropriate deterministic local algorithm (see \cref{sec:2lll_regime} for an example usage of the technique).
This technique, developed in several papers \cite{barenboim_elkin_pettie_schneider2016shattering,fischer_ghaffari2017sublogarithmic,chang_li_pettie2018optimal_coloring}, culminated in an algorithm by \citet{ghaffari2016MIS} that solves the problem in $O(\log \Delta + \poly\log\log n)$ rounds.

On the lower bound side, the aforementioned bound by \citet{kuhn16_jacm} shows that $O(\log \Delta)$ complexity cannot be substantially improved for certain ranges of $\Delta$. 

Moreover, the celebrated round elimination technique developed over the past decade \cite{brandt2019automatic,brandt_et_al2016_LLL_lower_bound,balliu2019LB_matching_MIS,balliu2021improved_mis,brandt2020truly,balliu2022hide_and_seek,balliu2020ruling} (see \cref{sec:2log_star_regime,sec:2lll_regime} or the survey by \citet*{suomela_survey}) shows that this complexity is nearly optimal. In particular, \citet*{balliu2019LB_matching_MIS} use round elimination to establish a lower bound of $\Omega\left( \min\left(
\Delta, 
\frac{\log \log n}{\log\log\log n}
\right)\right)$.
We conclude that no algorithm can achieve round complexity $O(\log \Delta) + o\left(\frac{\log \log n}{\log\log\log n}\right)$ or even $o(\Delta) + o\left(\frac{\log \log n}{\log\log\log n}\right)$.

Interestingly, a deterministic algorithm by Barenboim, Elkin, and Kuhn \cite{BEK15,barenboimelkin2009mis,Kuhn2009WeakColoring} for the maximal independent set problem runs in $O(\Delta + \log^* n)$ rounds, which is also tight under this lower bound.

\paragraph{Deterministic algorithms}
For deterministic algorithms, the most significant open question is whether their complexity can match that of Luby's algorithm.

\begin{problem}
    Is there a deterministic $\tilde{\Theta}(\log n)$-round algorithm for the maximal independent set problem?
\end{problem}

For a long time, the deterministic complexity of the maximal independent set problem paralleled the best-known algorithms for network decompositions (discussed in \cref{sec:1network_decomposition}), which implies a polylogarithmic round complexity. 
Recently, a number of papers developed techniques such as local rounding and distributed derandomization \cite{hanckowiak01,FischerGK17,harris2019distributed,ghaffari_kuhn2022deterministic_vertex_coloring,faour2023local,ghaffari_grunau2023faster_MIS_matching}; this lead to decoupling of the complexity of the maximal independent set from network decomposition. Currently, the best-known algorithm by \citet{ghaffari2024near} achieves a complexity of $\tilde{O}(\log^{5/3} n)$ rounds.

On the lower bound side, \citet*{balliu2019LB_matching_MIS} used round elimination to show that the deterministic local complexity of finding a maximal independent set is $\Omega(\log n / \log\log n)$.

%% file: chapter2/1log_star_regime.tex
\section{The Bounded-Degree Regime}
\label{chap:2_below_logn}

In the previous section, we gained a good understanding of the fundamentals of local complexity -- in essence, if we care about the complexities up to $\poly\log n$, there are several equivalent definitions of it with sequential local algorithms being a particularly helpful model for designing algorithms. 

In this section, we will restrict our attention to bounded degree graphs and local problems of sublogarithmic complexity. Something surprising is going to happen: We will see that while, a priori, we would expect all kinds of problem complexities, there are only three distinctive classes of local problems. This \emph{sharp threshold phenomenon} is a consequence of remarkable results known as \emph{speedup theorems}. 

As a rough roadmap for the rest of the section, we are going to give a more or less self-contained proof of the following \cref{thm:classification_basic}, with \cref{sec:2log_star_regime} focusing on the symmetry-breaking regime, \cref{sec:2lll_regime} focusing on the Lovász-local-lemma regime, and \cref{sec:2speedups_and_slowdowns} focusing on showing that there are gaps in between the regimes. 

\begin{restatable}[Classification of local problems with $o(\log n)$ complexity on bounded degree graphs]{theorem}{classification}
    \label{thm:classification_basic}
    Let us fix any $\Delta$ and the class of graphs of degree at most $\Delta$. Then, any local problem 
     with randomized round complexity $o(\log n)$ has one of the following three round complexities. 
    \begin{enumerate}
        \item \emph{Order-invariant regime:} The problem has $O(1)$ deterministic and randomized round complexity.   
        \item \emph{Symmetry-breaking regime:} The deterministic and randomized round complexity of the problem lies between $\Omega(\log\log^* n)$ and $O(\log^* n)$ (both the deterministic and the randomized complexity is the same function). 
        \item \emph{Lovász-local-lemma regime:} The problem has deterministic round complexity between $\Omega(\log n)$ and $\tilde{O}(\log^4 n)$. Its randomized round complexity is between $\Omega(\log\log n)$ and $\tilde{O}(\log^4\log n)$. 
    \end{enumerate}
\end{restatable}


\subsection{The Symmetry-Breaking Regime}
\label{sec:2log_star_regime}

In the bounded-degree regime, the problems of maximal independent set or the closely related problem of $(\Delta+1)$-coloring play a bit similar role to network decomposition, as we will see in \cref{thm:sequential_vs_distributed_coloring}. These problems are known as basic \emph{symmetry-breaking} problems. This section shows that the round complexity of these problems is $\Theta(\log^* n)$ on bounded-degree graphs. 

\tbox{
Basic symmetry breaking problems like maximal independent set and $\Delta+1$ coloring are closely related. Their round complexity on bounded degree graphs is $\Theta(\log^* n)$. 
}

\subsubsection{Fast Coloring Algorithm and its Implications}
We will first discuss a classical local coloring algorithm of \citet*{linial92LOCAL_definition} improving upon earlier work from \cite{cole86,goldberg_plotkin_shannon88early_log_star_coloring}. 
This algorithm colors a graph of degree at most $\Delta$ with $O(\Delta^2)$ many colors in $O(\log^* n)$ rounds. We will later see in \cref{thm:sequential_vs_distributed_coloring} that this implies that on bounded-degree graphs, the local complexity of constructing the maximal independent set and $\Delta+1$ coloring is $O(\log^* n)$. 

The main idea behind Linial's algorithm is as follows: At the beginning of the algorithm, we will think of the unique identifiers at the vertices as an input proper coloring with colors from the (very large) range $[n^{O(1)}]$. 
The heart of the proof is to provide a $1$-round local algorithm that takes as input a coloring with colors from a large range $[k]$, and outputs a coloring with colors from a much smaller range $[O(\Delta^2 \log k)]$. Repeating this process for $O(\log^* n)$ many rounds, we end up with a coloring with $\Delta^{O(1)}$ many colors. 

It remains to construct a suitable $1$-round algorithm. To do this, we will use the so-called cover-free families defined below. 

\begin{definition}
    \label{def:cover-free}
    Given a ground set $[k']$, a family of sets $S_1, \dots, S_k \subseteq [k']$ is called a $\Delta$-cover free family if for each set of indices $i_0, i_1, \dots, i_{\Delta} \in [k]$, we have $S_{i_0} \setminus \left( \bigcup_{j= 1}^\Delta S_{i_j} \right) \not= \emptyset$. 
That is, no set in the family is a subset of the union of some $\Delta$ other sets.
\end{definition}

Such families can be constructed with the following parameters. 
\begin{lemma}[{\cite{kleitman1981coverfree,erdos1982coverfree} or see \cite[Lemma 1.19, 1.20]{ghaffari2020DistributedGraphAlgorithms_lecturenotes}}]
\label{lem:coverfree_family}
    For any $k, \Delta$, there exists a $\Delta$-cover free family of size $k$ on a ground set of size $k' = O(\Delta^2 \log k)$. Moreover, if $k \le \Delta^3$, it exists for $k' = O(\Delta^2)$. 
\end{lemma}

We can now state and analyze Linial's algorithm. 

\begin{theorem}[\citet*{linial92LOCAL_definition}]
    \label{thm:coloring}
    The deterministic round complexity of constructing a coloring with $O(\Delta^2)$ many colors is $O(\log^* n)$. 
\end{theorem}
\begin{proof}
    We will show a one-round algorithm that turns an input proper $k$-coloring into a proper $O(\Delta^2 \log k)$ coloring. Moreover, in case $k \le \Delta^3$, we can turn the $k$-coloring into an $O(\Delta^2)$ coloring.

    The one-round algorithm simply interprets each input color from $[k]$ as a set in a $\Delta$-cover free family over the ground set $[k']$ for $k' = O(\Delta^2 \log k)$ (and $k' = O(\Delta^2)$ in case $k \le \Delta^3$). Such a family exists by \cref{lem:coverfree_family}. After each vertex $u$ with a color $S_u$ learns the colors $S_{v_1}, \dots, S_{v_d}$ of its neighbors, it simply chooses any color in the set $S_u \setminus \left( \bigcup_{i = 1}^d S_{v_i}\right)$ as its new color. This set is non-empty by the definition of a cover-free family, and the new coloring is proper. This finishes the description of our one-round algorithm. 
    
    We interpret the input unique identifiers as proper coloring and apply the one-round algorithm $O(\log^* n)$ many times. A quick calculation shows that the number of colors then drops to $O(\Delta^2 \log \Delta)$. After one more round and using the case $k \le \Delta^3$, the number of colors becomes $O(\Delta^2)$, as needed. \footnote{If we do not care about how the resulting number of colors depends on $\Delta$, there is a more self-contained one-round algorithm going back to \citet{cole86} that does not rely on cover-free families and leads to a coloring with $2^{O(\Delta\log \Delta)}$ many colors in $O(\log^* n)$ rounds. The algorithm turns a coloring with $2^k$ colors into a coloring with $2^{O(\Delta\log k)}$ colors as follows: Think of the input color of each vertex $u$ as a $k$-bit string $s_u$. The vertex sends this color to all its neighbors. After $u$ receives $d \le \Delta$ strings of its neighbors, $s_1, s_2, \dots, s_d$, it does the following. For each received $s_i$, the node $u$ identifies an index $j_i$ where $s_i$ differs from its own bit string $s_u$. Such an index exists since we assumed that we started with proper coloring. The new color of $u$ is defined as the concatenation $s'_u = j_1 \circ s_u[j_1] \circ \dots \circ j_d \circ s_{u}[j_d]$. That is, the vertex $u$ simply remembers only the values of bits of $s_u$ on the positions where those values reveal that $s_u$ differs from its neighbors. We observe that no two neighboring vertices can end up with the same string. Moreover, if we write $s'_u$ in binary, it has only $O(\Delta \cdot \log k)$ many bits. }
\end{proof}

\paragraph{Simulation of sequential local algorithms}
We notice that graph coloring can be seen as a special case of network decomposition discussed in \cref{chap:1_local_complexity_fundamentals} where each cluster has diameter $0$. In particular, coloring with a small number of colors allows us to turn sequential local algorithms into distributed local ones similarly to \cref{thm:sequential_vs_distributed_complexity}. 

\begin{theorem}\label{thm:sequential_vs_distributed_coloring}
    Let $\fA$ be a sequential local algorithm with local complexity $t(n)$. 
    Then, we can turn it into a distributed local algorithm of round complexity $O(\Delta^{O(t(n))} + t(n) \cdot \log^* n)$. If the sequential algorithm is deterministic, so is the distributed one. 

\end{theorem}
\begin{proof}
    We follow the outline of the proof of \cref{thm:sequential_vs_distributed_complexity} that simulated sequential local algorithms using network decompositions. 
    To simulate a sequential local algorithm $\fA$ of local complexity $t(n)$, we first construct a coloring of $G^{t(n)}$ with $O\left( \Delta^2(G^{t(n)}) \right)$ colors\footnote{The notation $\Delta(G)$ stands for the maximum degree of $G$. } using Linial's algorithm from \cref{thm:coloring}. 
    That algorithm needs $O\left(t(n) \cdot \log^* n\right)$ rounds where we multiply by $t(n)$ because one round of communication in $G^{t(n)}$ can be simulated in $t(n)$ rounds in $G$. 
    Subsequently, we iterate over all colors and simulate the sequential algorithm $\fA$ as in \cref{thm:sequential_vs_distributed_complexity}. That simulation takes additional $O(\Delta^2(G^{t(n)})) = \Delta^{O(t(n))}$ rounds. 
\end{proof}

As a corollary of \cref{thm:sequential_vs_distributed_coloring}, we get that maximal independent set or $\Delta+1$-coloring can be constructed in $O(\log^* n)$ rounds on bounded-degree graphs. 

\begin{corollary}
\label{cor:mis_coloring}
    Sequential local algorithms with local complexity $O(1)$ (such as algorithms for the maximal independent set or $\Delta+1$ coloring) can be simulated with round complexity $O(\log^* n)$ on bounded degree graphs.  
\end{corollary}

\paragraph{Understanding unique identifiers}
Let us now discuss the unique identifiers from the range $[n^{O(1)}]$ in the definition of deterministic algorithms. We will use our understanding of coloring to see that the strength of the model of deterministic algorithms remains the same even if the identifiers come from a much larger range like $[2^n]$ or $[2^{2^n}]$. Moreover, we will understand why identifiers are not needed in the definition of sequential local algorithms. 

The following theorem will use an instance of a \emph{fooling argument}, variants of which we will enjoy employing later on. To motivate it, notice that it is a bit awkward that the definition of a \emph{local} algorithm talks about \emph{globally} unique identifiers. 
We will next ``fool'' a given deterministic local algorithm solving a local problem by supplying to it a distance coloring (i.e., coloring of the power graph) with $n^{O(1)}$ many colors instead of unique identifiers. The algorithm still has to work since a failure of the algorithm for input distance coloring would imply a failure for input unique identifiers.

\begin{theorem}
    \label{thm:bad_id_into_good_id}
    Let $\fA$ be a deterministic local algorithm of round complexity $t(n)$ for a local problem $\Pi$. Then, given any $s = n^{\omega(1)}$, there is a deterministic local algorithm $\fA'$ solving $\Pi$ in $O(t(n) \cdot \log^*_n(s))$ rounds\footnote{Here, $\log^*_n(s)$ returns how many times we need to take the logarithm of $s$, until its value drops below $n$, i.e., $\log^*_n(2^n) = 1, \log^*_n(2^{2^n}) = 2$ and so on. } and assumes that the identifiers are from range $[s]$. 

    Similarly, let $\fA$ be a sequential local algorithm of local complexity $t(n)$ for a local problem $\Pi$ in a model of sequential local algorithms where each node has additionally a unique identifier from $[n^{O(1)}]$. Then there is a sequential local algorithm $\fA'$ for $\Pi$ of complexity $O(t(n))$ in the standard model of sequential local algorithms without any identifiers. 
\end{theorem}
\begin{proof}
Let $\fA$ be a deterministic algorithm of round complexity $t(n)$ solving a local problem $\Pi$ with local checkability $r$ in the model where the unique identifiers are from $[n^{O(1)}]$ (the constant in the exponent is assumed to be large enough). We construct a new algorithm $\fA'$ that works in the less powerful model with identifiers from $[s]$ as follows. We first compute a coloring of the power graph $G^{2(t(n) + r)}$ with $n^{O(1)}$ many colors, then we simulate $\fA$ with that coloring as identifiers. 

In particular, the coloring is constructed by iterating color reductions of Linial's algorithm of \cref{thm:coloring}. Recall that each color reduction reduces the range of color exponentially, thus after $\log^*_n(s)$ rounds, we reduce the input coloring with colors from $[s]$ to a coloring with colors of size at most  $O(\Delta^2(G^{2(t(n) + r)})) = n^{O(1)}$. 


Next, we prove that $\fA'$ is correct. Assume that $\fA'$ fails to solve $\Pi$ at a node $u$. Notice that this failure depends only on the $(t(n) + r)$-hop neighborhood of $u$ where the coloring constructed by $\fA'$ uses unique colors. In particular, we can extend this coloring of $B(u, t(n)+r)$ to a labeling of every node of the input graph with unique identifiers. The original algorithm $\fA$ fails to solve $\Pi$ at $u$ for these identifiers, a contradiction with $\fA$ being correct. 

The proof of the second claim in the theorem is very similar and, in fact, easier, since $\Delta+1$-coloring has constant sequential local complexity and thus Linial's color reductions are not needed. 
\end{proof}

As a helpful corollary of the second part of \cref{thm:bad_id_into_good_id}, we can now see that any deterministic local algorithm $\fA$ solving some local problem can be converted into a deterministic sequential algorithm of the same asymptotic local complexity. 
This was actually not clear until now since our definition of deterministic sequential algorithms in \cref{def:sequential_algorithm} did not contain input unique identifiers.

\subsubsection{Lower bound for coloring via round elimination}
Finally, let us prove that the $\log^* n$ dependency for constructing maximal independent sets or colorings is necessary. 

\begin{theorem}[\citet*{naor1991lower_bound_ring_coloring,linial92LOCAL_definition}]
\label{thm:coloring_lower_bound}
    The local complexity of computing $\Delta+1$ coloring is $\Omega(\log^* n)$, even on graphs that are oriented paths.  
\end{theorem}

There are several known proofs of this theorem \cite{naor1991lower_bound_ring_coloring,linial92LOCAL_definition}; we will use the proof of \citet*{linial92LOCAL_definition} framed in the language of a powerful technique known as the \emph{round elimination} (see \cite{suomela_survey} for an introduction to it). 
In essence, given a local problem $\Pi$, round elimination is an automated technique that defines a problem $\Pi'$ such that the round complexity of the fastest algorithm for $\Pi'$ is exactly one round smaller than the round complexity of the fastest algorithm for $\Pi$ (unless the complexity of $\Pi$ was already zero). 
This is very helpful for proving lower bounds: If we start with some problem $\Pi$ and argue that even after $t$ rounds of round elimination, the problem $\Pi^{(t)}$ that we end up with is not solvable in zero rounds, we infer that the local complexity of $\Pi$ is at least $t$. 


\paragraph{Preparations for the lower bound}
We will prove the lower bound in the extremely simple setup where we are promised that the input graph is an oriented path, i.e., the setup from \cref{sec:1first_example}. The lower bound will be for deterministic algorithms, so each node starts with a unique identifier. We will make a further restriction on the identifiers, we require that the input labeling with identifiers is increasing; that is, if for every oriented edge $e = uv$ going from $u$ to $v$, we have $ID(u) < ID(v)$. 

Next, let us discuss local algorithms. It will be helpful to think about them as functions in the spirit of \cref{def:local_algorithm}. 
A subtlety we need to be careful about is that in one step of round elimination, we don't want to directly convert a $t$-round local algorithm (that sees $2t+1$ vertices) to a $t-1$ round algorithm (that sees $2t-1$ vertices). Instead, we want to convert an algorithm that sees $t$ vertices to an algorithm that sees $t-1$ vertices. 
To this end, we define an \emph{edge-centered} $(t+1/2)$-round algorithm $\fA$ to be a local algorithm such that for an edge $e = uv$, the input to $\fA$ is the ball $B(uv, t-1)$ defined as $B(u, t-1) \cup B(v, t-1)$. The output of $\fA$ is a label for the edge $e$. For example, $1/2$-round local algorithm run on $e$ has access to the two identifiers $ID(u)$ and $ID(v)$ and it maps the two identifiers to a label of $e$. 

\paragraph{One half-round reduction}
Here comes the heart of the proof; we will show that any given $t$-round local algorithm for coloring vertices with $k$ colors can be converted to a $((t-1/2))$-round algorithm for coloring edges with $2^k$ colors.  


\begin{lemma}
    \label{lem:coloring_elimination}
    Assume that for $t > 0$ we are given a node-centered $t$-round deterministic local algorithm $\fA$ that outputs a proper coloring with $k$ colors on oriented paths with increasing unique identifiers. 
    Then, there is an edge-centered $((t-1/2))$-round algorithm $\fA'$ that properly colors the edges with $2^k$ colors in the same setup. 
    Similarly, edge-centered $(t+1/2)$-round algorithms for $k$-coloring can be converted into node-centered $t$-round algorithms for $2^k$-coloring. 
\end{lemma}
\begin{proof}
    We will only discuss the conversion of a node-centered algorithm into an edge-centered one, the other case is very similar. 
    
    Given a $t$-round algorithm $\fA$, we define a $(t-1/2)$-round edge-centered algorithm $\fA'$ as follows (see \cref{fig:elimination1}): 
    Given an edge $e = uv$ so that $\fA'$ has access to $B(uv, t-1)$, we let $\fA'$ to consider all possible identifiers of the unique node $x \in B(v, t) \setminus  B(uv, t-1)$, i.e., the only node that $\fA$ sees but $\fA'$ doesn't. 
    For each possible value of the identifier $ID(x)$ of $x$ (we also consider the option that we are at the end of the path and the vertex $x$ does not exist), the algorithm $\fA'$ simulates $\fA$ on $B(v,t)$ and records the color that $\fA$ outputs. 
    The final output of $\fA'$ is a subset of $[k]$ that contains each color $c \in [k]$ whenever there is an identifier that makes $\fA$ output $c$. 
    Note that we can encode the set with $k$ bits -- that is, we view this set of colors from the range $[k]$ itself as a new color from the range $[2^k]$. This finishes the description of $\fA'$. 

    \begin{figure}
        \centering
        \includegraphics[width = \textwidth]{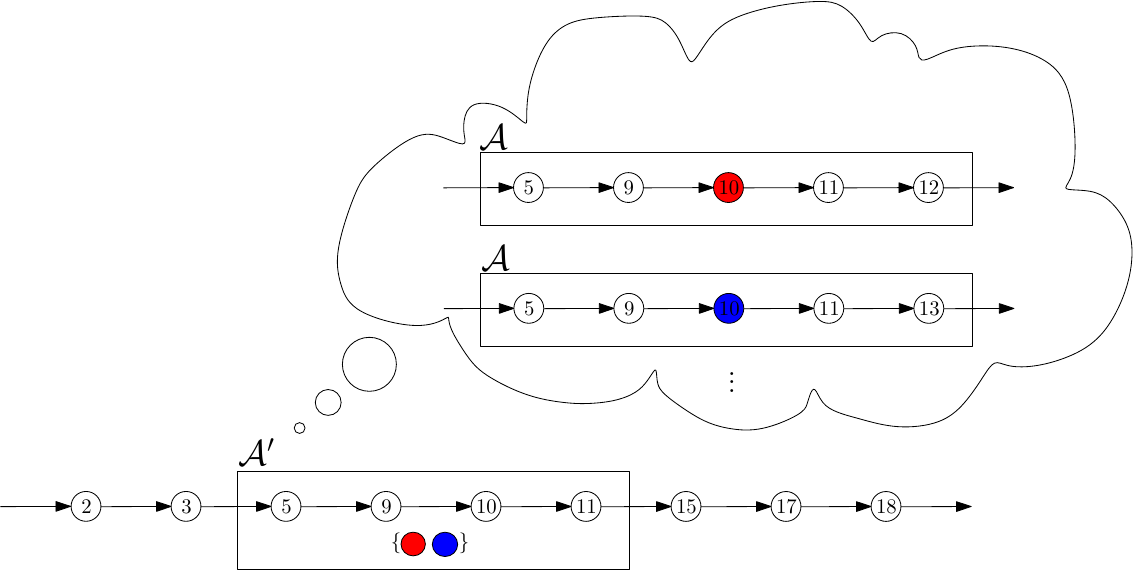}
        \caption{The definition of the $1.5$-round algorithm $\fA'$ derived from a $2$-round algorithm $\fA$: The algorithm $\fA'$ considers all compatible one-node extensions of its neighborhood containing $4$ nodes to $5$-node-neighborhoods (in the picture, this corresponds to the following node having the identifier $12, 13, \dots, n^{O(1)}$, and the possibility that the node does not exist). The algorithm considers all possible answers that $\fA$ returns for those neighborhoods (the picture shows that the red and the blue color are two of those possible answers). 
        The color that $\fA'$ uses to color the edge it is centered on is simply the set of all colors that $\fA$ returns for some extension (in the picture, it is the set containing the red and the blue color).  }
        \label{fig:elimination1}
    \end{figure}

    Our task is to prove that $\fA'$ returns a proper coloring of edges. Here is what this question reduces to: Consider any increasing labeling of the oriented path with identifiers and consider any three consecutive nodes $u,v,w$ (see \cref{fig:elimination2}). 
    We again let $x$ be the unique node in $B(v,t) \setminus B(uv, t-1)$ and $y$ the subsequent neighbor of $x$. 
    Consider the new colors $C_1, C_2 \in [2^k]$ that $\fA'$ outputs on the two edges $uv$ and $vw$. 
    We need to prove that the two colors $C_1, C_2$ are different. We do that by focusing on the original color $c \in [k]$ that $\fA$ outputs at $v$. 

    \begin{figure}
        \centering
        \includegraphics[width = \textwidth]{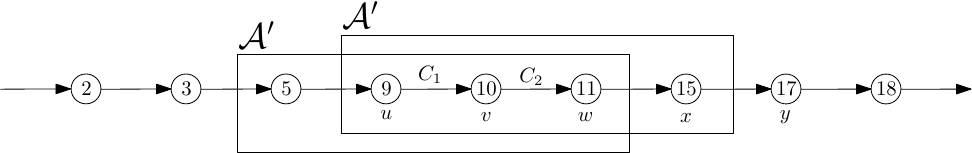}
        \caption{The situation from the proof of correctness of $\fA'$: We need to argue that the two colors $C_1, C_2$ are different. To do so, we define $c$ to be the color that $\fA$ outputs on the input $(5, 9, 10, 11, 15)$. By definition, we have $c \in C_1$. On the other hand, if $c \in C_2$, we get an existence of an identifier $ID'(y)$ such that $\fA$ returns $c$ on the input $(9, 10, 11, 15, ID'(y))$. But then we consider the input $(5, 9, 10, 11, 15, ID'(y))$ and notice that $\fA$ colors two consecutive vertices with the same color $c$, a contradiction with its correctness. }
        \label{fig:elimination2}
    \end{figure}

    On one hand, notice that $c \in C_1$ because the edge-centered algorithm $\fA'$ run at $uv$ considers the actual identifier $ID(x)$ as a possibility and thus includes $c$ in $C_1$. 
    
    On the other hand, assume for contradiction that $c \in C_2$. That would imply that there is a certain identifier $ID'(y)$ that makes the output of $\fA$ at $w$ to be the color $c$. But now consider changing the valid increasing labeling of identifiers that we started with by letting the identifier of $y$ be equal to $ID'(y)$. Let's look at all the nodes in $B(vw,t)$. The identifiers on them are still a valid increasing identifier sequence.\footnote{This is the place in the argument where we need increasing and not just unique identifiers: With unique identifiers, the leftmost and the rightmost node of $B(vw, t)$ could now have the same identifier. } 
    But after extending $B(vw,t)$ and its identifiers to a path on $n$ vertices labeled with increasing identifiers, $\fA$ fails to solve $k$-coloring on that graph since it outputs the same color $c$ at both $v$ and $w$, a contradiction with our assumption that $\fA$ is correct. 
\end{proof}

We can now prove \cref{thm:coloring_lower_bound} as follows: Assuming the existence of a $t = o(\log^* n)$-round algorithm for $3$-coloring of oriented paths, we apply \cref{lem:coloring_elimination} $2t$ times until we end up with a $0$-round algorithm $\fA_0$ that colors oriented paths with $2^{2^{\iddots^{2^3}}} < n$ colors where the inequality holds because the height of the exponentiation tower is $2t = o(\log^* n)$. 
But such $\fA_0$ is simply a function mapping an input identifier from the range $[n^{O(1)}]$ to a color in the smaller range $[n]$. Using the pigeonhole principle, we can find two identifiers that $\fA_0$ maps to the same color and argue that if these two identifiers happen to be present at two neighboring nodes, $\fA_0$ fails to output proper coloring. \footnote{Round elimination may used both to prove lower bounds \emph{and} to construct algorithms. In particular, it can be used to derive that the round complexity of $3$-coloring paths in $\frac12 \log^* n \pm O(1)$ \cite{rybicki_suomela2015exact_bounds_for_coloring}. }

%% file: chapter2/2lll_regime.tex
\subsection{The Lovász Local Lemma Regime}
\label{sec:2lll_regime}

We will next discuss Lovász local lemma, a very expressible local problem closely related to the third regime of problems from \cref{thm:classification_basic}. 

\paragraph{Definition of Lovász local lemma}
Lovász local lemma is the following very general local problem \footnote{The Lovász-local-lemma problem does not quite satisfy the requirements for the local problem as we defined it in \cref{def:local_problem}, but, morally speaking, it can be seen as such. }. The problem is formally defined for bipartite graphs in which nodes of one part are labeled as \emph{random variables} and the nodes of the other part are labeled as \emph{bad events}. We denote the maximum degree of random variable nodes as $\Delta_{r.v.}$ and the maximum degree of bad event nodes as $\Delta_{b.e.}$. Each node $u$ labeled as a random variable is additionally labeled with a probability space $\Omega_u$. To simplify discussions, we will without loss of generality assume that each probability space is an infinite list of random bits (i.e., it is the uniform distribution on $[0,1]$). 
Next, each node $v$ labeled as a bad event that neighbors with nodes $u_1, \dots, u_d$ with $d \le {\Delta_{b.e.}}$ is additionally labeled with an event $\fE_v$ on the space $\prod_{i = 1}^d \Omega_{u_i}$. 

Often, it is useful to work only with a graph induced by bad-event nodes where two bad-event nodes are connected if they share a common random variable. In that case, we will talk about the \emph{dependency-graph} formulation of Lovász local lemma and use $\Delta$ to denote its maximum degree. On the other hand, the setup with a bipartite graph with both bad-event nodes and random-variable nodes will be denoted as the \emph{variable-event-graph}. 

The crucial ingredient to the Lovász local lemma is a requirement on the bad events that, roughly speaking, says that we can use the union bound in the dependency-graph neighborhood of every bad event $\fE_v$ and conclude that with positive probability, neither $\fE_v$, nor its neighboring bad events occur. 
Namely, in the \emph{tight} version of Lovász local lemma, we are given a promise 
 that each bad event $\fE_v$ has probability at most $p$ where $p$ is defined by the following \emph{Lovász local lemma criterion}:  
    \begin{align}
        \label{eq:lll_classical_definition}
        p \cdot (\Delta + 1) \le 1/\e. 
    \end{align}

A foundational result of \citet*{erdos_lovasz1975lovasz_local_lemma} is that even in the tight formulation, it is always possible to solve any instance of the local lemma problem, by which we mean that one can instantiate random variables so that no bad event occurs. 

A \emph{$C$-relaxed} (or just polynomially-relaxed) version of Lovász local lemma, which is a bit more relevant to our applications, only requires that 
\begin{align}
    \label{eq:lll_definition}
    p \cdot \Delta^C \le 1
\end{align}
This section will show both fast algorithms and lower bounds for the polynomially-relaxed Lovász local lemma. 

\tbox{
Lovász local lemma is a very versatile and expressible local problem. Its deterministic round complexity on bounded degree graphs is $\Theta(\poly\log n)$, and its randomized complexity $\Theta(\poly\log\log n)$. 
}

\subsubsection{Fast Algorithms for the Local Lemma}

We first discuss how to solve any instance of the local lemma with a fast deterministic local algorithm. To this end, we will start with a randomized algorithm that we later derandomize by \cref{thm:derandomization}. In a breakthrough result in the area of constructive algorithms for the local lemma, \citet*{moser2010constructiveLLL} presented an algorithm for it in the tight formulation. Moreover, they presented a parallel variant of their algorithm that can be interpreted as a local algorithm with round complexity $O(\log^2 n)$. This complexity was later improved by \citet*{chung2017LLL} to $O(\log n)$ rounds on bounded degree graphs. Plugging their result to \cref{thm:derandomization}, we obtain the following theorem. 

\begin{theorem}
\label{cor:deterministic_lll}
    The deterministic round complexity of solving any instance of the tight version of Lovász local lemma on bounded degree graphs is $\tilde{O}\left( \log^4(n) \right)$. 
\end{theorem}

\paragraph{Even faster randomized algorithm}
Next, we will show that there is an even exponentially faster randomized algorithm for the relaxed version of the local lemma. 
We will discuss the algorithm of \citet*{fischer_ghaffari2017sublogarithmic} with round complexity $\poly\log\log n$ which uses \emph{shattering}, a successful technique (cf. \cref{sec:3concrete_problems}) that goes back to the work on algorithmic local lemma by \citet*{beck1991algorithmic}. 

The main idea goes as follows. Our algorithm will have two phases. In the first phase, we use a sequential local algorithm with constant complexity to do the following. 
Given an instance of the local lemma, the algorithm fixes \emph{most} of the random variables in such a way that \emph{most} bad events are satisfied (i.e., all random variables relevant for that bad event are fixed and the event does not occur). 
More concretely, the fraction of fixed random variables and satisfied events is $1 - 1/\Delta^{O(1)}$.  
The price for this outcome is that remaining, unfixed, bad events have their probability slightly increased from $1/\Delta^C$ to $1/\Delta^{C'}$ for some $C' < C$. 

Fortunately, we have an additional guarantee on the unfixed bad events. 
Due to the very small locality of the algorithm in the first phase, we can use an independence-like argument to show that the size of the connected components of unfixed bad events is $O(\Delta^{O(1)} \log n)$, with high probability (see \cref{fig:lll-example}). 
We also say that the graph \emph{shatters} into small components.  
We can thus run the best deterministic algorithm from \cref{cor:deterministic_lll} as the second phase of the algorithm, to solve the remaining instance of the $C'$-relaxed local lemma. 
This second phase takes $\poly\log\left( \Delta \log n\right)$ rounds which is also the round complexity of the overall algorithm.\footnote{Our example algorithm from \cref{sec:1first_example} can be seen as a simple shattering algorithm. }
We will now make the above discussion formal. 

\begin{figure}
    \centering
    \includegraphics{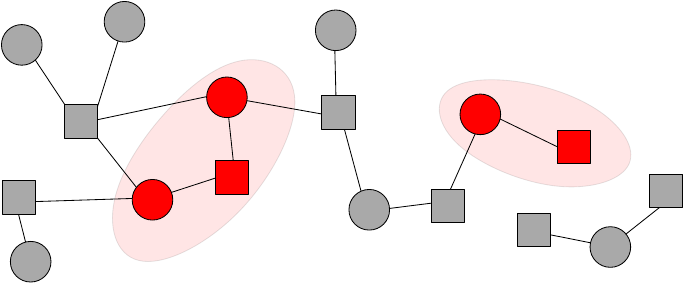}
    \caption{This picture shows an example variable-event graph corresponding to an instance of Lovász local lemma; the circles represent bad events and squares represent random variables. 
    The picture shows the situation after the first phase of Fischer-Ghaffari algorithm (\cref{thm:fischer_ghaffari_first_phase}). Most random variables are set to a fixed value (grey squares). A small proportion of the random variables remains unset (red squares) and the bad events neighboring with an unset random variable (red circles) form connected components of diameter $O(\log n)$.  }
    \label{fig:lll-example}
\end{figure}

\begin{theorem}[First phase of \citet*{fischer_ghaffari2017sublogarithmic}]
\label{thm:fischer_ghaffari_first_phase}
    There is a randomized sequential local algorithm $\fA$ of constant local complexity that gets as input a variable-event graph of an instance of $C$-relaxed Lovász local lemma for some large enough $C$. The algorithm fixes some bits of each random variable to concrete values so that conditioned on those fixed bits, we have the following two properties: 
    \begin{enumerate}
        \item Each bad event has probability at most $1/(3\Delta)$. 
        \item Up to $1/\poly(n)$ error probability, the residual dependency graph induced by bad events with non-zero probability has connected components of size $O(\Delta^3 \log n)$. 
    \end{enumerate}
\end{theorem}
\begin{proof}\footnote{The original algorithm and its analysis by \citet*{fischer_ghaffari2017sublogarithmic} contains subtle flaws. Here, we follow the proof by \citet*{yannic}.  }
The algorithm $\fA$ is defined as follows. We iterate over the random variables and for each random-variable node $u$ with neighboring bad-event nodes $v_1, \dots, v_d$, we perform the following process. 
We sample the random bits of $\Omega_u$ one by one. While we do that, we consider the probabilities of neighboring bad events $P(\fE_i | \text{sampled bits})$ for each $1 \le i \le d$. 
The first time it happens that for some $i$ we have $P(\fE_i| \text{sampled bits}) > 1/(6\Delta)$, i.e., the bad-event probability crosses the \emph{dangerous threshold} of $1/(6\Delta)$, we stop sampling and do the following. 

We label $u$, together with all other random variables neighboring $v_i$ as \emph{frozen} (if more $v_i$'s jumped over the dangerous threshold at the same time, we do this to all of them). We never sample bits of frozen random variables in the future. In particular, we stop sampling bits of $u$ and continue with the next unfrozen random variable in the arbitrary order of our sequential local algorithm. This finishes the description of the algorithm. 

\vspace{1em}

To see that each bad event has probability at most $1/(3\Delta)$ (the first item in the theorem statement) during any point throughout the algorithm, we notice that if an event $\fE$ of probability $P(\fE) = p$ depends on a single bit of randomness $b$, we have, by Markov's inequality, that $P(\fE | b = 0), P(\fE | b = 1) \le 2p$. That is, the bad event probability in our process never increases multiplicatively by a factor larger than $2$, which together with the definition of the dangerous threshold implies the desired bound. 

To understand the size of connected components in the residual dependency graph (i.e., components of bad-event nodes of non-zero probability after we set the random bits), consider any fixed set $S$ of bad-event nodes with the following two properties:
\begin{enumerate}
    \item (\emph{independence}) no two nodes $v_1, v_2 \in S$ are neighboring in the dependency graph $G$,
    \item (\emph{connectivity}) $S$ is connected in $G^{6}$. 
\end{enumerate}
We will next prove that with high probability, no such set $S$ of size $\Omega(\log n)$ survives to the residual graph. 

First, we use the independence property to prove that the probability that all nodes in a fixed set $S$ crossed the dangerous threshold during our process is exponentially small in the number of nodes $|S|$. 
Let us contemplate the behavior of the algorithm $\fA$ with respect to any bad-event node $v \in S$ and the random variables $u_1, \dots, u_d$ neighboring $v$ in the variable-event graph. 
We note that 
if bits $b_1, \dots, b_{k-1}$ were already sampled and $b_k$ is sampled next, we have $E_{b_k}\left[ P(\fE_{v} | b_1, \dots, b_{k})\right] = P(\fE_{v} | b_1, \dots, b_{k-1})$. Hence, viewing $P(\fE_v)$ as a random variable that depends on bits sampled from $\Omega_u, u \in V(G)$, we conclude that its expectation is at most $1/\Delta^C$ and we can thus use Markov's inequality to conclude that the probability of $P(\fE_{v})$ crossing the dangerous threshold is at most $\frac{6\Delta}{\Delta^C} = 6/\Delta^{C-1}$. 
Moreover, we notice that if we first set the randomness of nodes in $V' = V(G) \setminus (u_1 \cup \dots \cup u_d)$ to whatever values, we can still make above argument and conclude that for all ways of setting $\Omega_u = \omega_u$ for all $u \in V'$ we have
\[
P(\fE_{v} | \forall u \in V': \Omega_u = \omega_u ) \le 6/\Delta^{C-1}. 
\]
Since we assumed that $S$ is an independent set of bad-event nodes in the dependency graph, 
we can thus inductively prove that the probability of all nodes in $S$ crossing the dangerous threshold is at most $\left( \frac{6}{\Delta^{C-1}} \right)^{|S|}$.  

We next count the number of possible sets $S$ of size $t$ that satisfy the connectivity property: Each such $S$ can be specified by fixing any node $u \in S$, and then specifying how one can walk for $2(|S|-1)$ steps in $G^6$ so that the walk defines a spanning tree of $G^6[S]$. Hence, the number of sets $S$ of size $t$ is at most $n \cdot (2\Delta^6)^{2(t-1)}$. We can thus upper bound the existence of some connected surviving set $S$ of size $t$ as 
$$
n \cdot (2\Delta)^{12t} \cdot \left(\frac{6}{\Delta^{C-1}}\right)^{t}.
$$
Choosing $C = O(1)$ and $t = O(\log n)$ large enough, we conclude that the size of this expression is at most $1/n^{O(1)}$. 

\vspace{1em}

Finally, for $t = O(\log n)$ from the above argument, consider any connected subset $S$ of $G$ of size at least $(2\Delta)^3 \cdot t$ and assume that $S$ survived to the residual dependency graph. 
Then, we can find its subset $S' \subseteq S$ of size at least $|S'| \ge t$ where $S'$ is independent in $G^3$, yet $S'$ is connected in $G^{4}$. 
To see this, consider a greedy algorithm that starts with $S' = \{u\}$ for arbitrary $u \in S$, iterates through vertices of $S$ and while we still can add at least one node $v \in S$ to $S'$ and keep the independence property, we choose any $v$ with distance $4$ to $S'$ and add it to $S'$. One can see that the final set $S'$ has to be connected in $G^{4}$. Also, $|S'| \ge t$ because adding a vertex to $S'$ disqualifies at most $(2\Delta^3)$ vertices to be added to $S'$ in the future. 

If all nodes of $S$ survived to the residual dependency graph, it means that every node in $S'$ has to have a neighboring node that crossed the dangerous threshold. The set $S'$ thus gives rise to a set $S''$ of size $|S''| = |S'| \ge t$ of nodes that all crossed the dangerous threshold. Moreover, by $S'$ being independent in $G^3$, we conclude that $S''$ is independent in $G$. Since $S'$ was connected in $G^4$, $S''$ is connected in $G^6$. The set $S''$ thus satisfies the requirements of a set that, as we have proven, does not occur in the residual graph with high probability and we can thus conclude the same for the original connected set $S$. 
\end{proof}

Putting \cref{thm:fischer_ghaffari_first_phase} (simulated as a distributed algorithm by \cref{thm:sequential_vs_distributed_coloring}) together with \cref{cor:deterministic_lll}, we conclude that the following result holds. 

\begin{theorem}[\citet*{fischer_ghaffari2017sublogarithmic}]
\label{thm:fischer_ghaffari}
    There exists a constant $C$ such that the randomized round complexity of any instance of the $C$-relaxed Lovász local lemma on bounded-degree graphs is $\tilde{O}\left( \log^4  \log n\right)$. 
\end{theorem}
The dependency on $\Delta$ in the round complexity of Fischer-Ghaffari algorithm is polynomial. It is unknown whether this polynomial dependency can be improved to logarithmic. 
\begin{problem}
    \label{prob:lll}
    Is there a randomized local algorithm for relaxed Lovász local lemma with round complexity $O(\poly\log\Delta + \poly\log\log n)$? 
\end{problem}

\paragraph{Conjecture of \citet*{chang_pettie2019time_hierarchy_trees_rand_speedup}}
We note that \citet*[Conjecture 1]{chang_pettie2019time_hierarchy_trees_rand_speedup} conjecture that any instance of the relaxed local lemma can be solved with randomized $O(\log\log n)$ complexity. 
\begin{problem}[Chang-Pettie Conjecture]
\label{conj:chang_pettie}
Is it true that the randomized round complexity of $C$-relaxed Lovász Local Lemma for some large enough $C$ is $\Theta(\log \log n)$ on bounded degree graphs? 
\end{problem}
We will see later in \cref{subsec:speedups} that the randomized round complexity $O(\log\log n)$ implies the deterministic round complexity of $O(\log n)$ via \cref{thm:exponential_derandomization}. 

\paragraph{Self-contained algorithm} A bit unfortunate property of \cref{thm:fischer_ghaffari} is that it relies on the randomized  entropy-compression algorithms from \citet*{moser2010constructiveLLL} or \citet*{chung2017LLL}. However, we can prove \cref{thm:fischer_ghaffari} also in a self-contained way: First, we observe that after the first phase of the Fischer-Ghaffari algorithm finished, each surviving bad event simply collects the whole connected component of the residual graph; then we solve the problem in each residual component by applying the standard, existential, Lovász local lemma. 
This new randomized algorithm has complexity $O(\log n)$ on bounded degree graphs since this is the maximum diameter of residual components, with high probability. 
We can derandomize this algorithm using \cref{thm:derandomization} and use the resulting deterministic algorithm instead of the algorithm of \citet*{chung2017LLL} in the second phase of our shattering algorithm. This way, we get a simpler and more self-contained algorithm. Its downside is a worse value of $C$ and a worse dependence of round complexity on $\Delta$. 

\subsubsection{Lower Bound via Round Elimination}

Next, we will show that the deterministic local complexity of solving a certain specific instance of Lovász local lemma is $\Omega(\log n)$. Later, we will prove in \cref{thm:exponential_derandomization} that this also implies a randomized lower bound of $\Omega(\log\log n)$, thus showing that the round complexity of Fischer-Ghaffari algorithm from \cref{thm:fischer_ghaffari} is close to tight. 

The problem we choose for the lower bound is \emph{sinkless orientation}. We will define the problem only on trees of degree at most $\Delta$ where it is already hard. The task is to orient all the edges of the input tree so that no node is a \emph{sink}, which is defined as a node such that all $\Delta$ neighboring edges point towards it (nodes of degree less than $\Delta$ are not constrained in any way). 

Sinkless orientation can be seen as a specific instance of the local lemma: orienting each edge randomly corresponds to a random variable at each edge. Then, a vertex becoming a sink corresponds to a bad event of probability $2^{-\Delta}$. For large enough $\Delta$, this is much smaller than the polynomial criterion $1/\Delta^{C}$ from \cref{eq:lll_definition} which makes this problem a valid instance of the local lemma. 
We will next prove the following theorem. 
\begin{theorem}[\citet*{brandt_et_al2016_LLL_lower_bound}]
    \label{thm:sinkless_orientation_lb}
For any constant $\Delta$, the sinkless orientation problem has deterministic round complexity $\Omega(\log n)$ on the class of trees of degree at most $\Delta$. 
\end{theorem}

\paragraph{Preparations}
As in \cref{thm:coloring_lower_bound}, we will want to replace the uniqueness of identifiers with local constraints that imply that they are unique. In the proof of \cref{thm:coloring_lower_bound}, we worked with identifiers that were monotonically increasing (which implied they were unique), this time we will work with identifiers that are consistent with a so-called \emph{ID graph} 
\cite{korhonen_paz_rybicki_schmid_suomela2021supported_model,balliu2023sinkless_orientation_made_simple,brandt_chang_grebik_grunau_rozhon_vidnyanszky2021trees}. \footnote{The usual usage of round elimination is for randomized algorithms instead of using the ID graph, but this would require a longer setup and more calculations. } 

\begin{definition}[ID graph]
\label{def:id_graph}
    Given a parameter $\Delta$, an ID graph $H$ is a graph on a set $[n]$ that we associate with unique identifiers. Every edge of the graph is colored with one of $\Delta$ many colors and we write $H_i$ for the graph induced by the $i$-th color. We require that: 
    \begin{enumerate}
        \item The girth of $H$, i.e., the length of the shortest cycle in $H$, is at least $\gamma \log_{\Delta} n$ for some fixed $\gamma > 0$.
        \item Each independent set of each $H_i$ has less than $n/ \Delta$ vertices.  
    \end{enumerate}
\end{definition}
Such a graph exists, which can be proven using the same argument as how one proves that high-girth high-chromatic graphs exist \cite{brandt_chang_grebik_grunau_rozhon_vidnyanszky2021trees,korhonen_paz_rybicki_schmid_suomela2021supported_model}.

We will fix any ID graph and work in a model \emph{relative} to that ID graph. Here is what that means. First, we will assume that the input graph $G$ is always a tree of degree at most $\Delta$ such that its edges are, moreover, properly $\Delta$-colored on the input (i.e., each vertex is incident to edges of different colors). Additionally, we require that if two nodes $u,v$ neighbor in $G$ with an edge of color $i$, then their identifiers $ID(u), ID(v)$ neighbor in $H_i$. 

We notice that this is a local constraint on input identifiers that does not imply they are unique. However, notice that whenever two vertices $u,v \in V(G)$ have the same identifier, we can consider the path between $u$ and $v$ in $G$ and how it maps to a walk in $H$ that starts in $ID(u)$ and finishes in $ID(v) = ID(u)$. The proper edge-coloring of $G$ implies that the walk never goes from $x \in V(H)$ to $y \in V(H)$ and then back to $x$ in the subsequent step. This implies that the walk contains a cycle, thus the distance of $u$ and $v$ is at least as large as the girth of $H$. That is, the identifiers are unique up to a large distance which is pretty much the same as them being unique (cf. \cref{thm:bad_id_into_good_id}). 

Similarly to the lower bound of \cref{thm:coloring_lower_bound}, we will work with node-centered and edge-centered algorithms. For node-centered algorithms, solving sinkless orientation means that the algorithm outputs one of $\Delta$ many input edge colors at each node. Outputting a color $i$ means that $u$ decides that the edge of the color $i$ goes outwards from $u$. The local constraint on this output vertex-coloring is that no two neighboring nodes should select the same edge going in between. We will call this variant of the problem \emph{edge-grabbing}. On the other hand, edge-centered algorithms will solve the sinkless orientation as we described the problem: Each edge simply outputs how it is oriented and we have a local constraint at each vertex, requiring that it has at least one outgoing edge. 

\paragraph{Round elimination}
We proceed with performing the actual round elimination. Notice that the spirit of the following proof is very similar to the proof of \cref{lem:coloring_elimination}. 

\begin{lemma}
\label{lem:sinkless_orientation_round_elimination}
    Assume that we are given a node-centered deterministic $t$-round local algorithm $\fA$ of round complexity at most $t \le \log_{\Delta} n \,-1$ that solves the edge-grabbing problem on trees of degree at most $\Delta$ relative to some fixed ID graph $H$. Then, there is a $t-1/2$-round edge-centered deterministic local algorithm $\fA'$ that solves sinkless orientation in the same setup. 

    Similarly, $t+1/2$-round edge-centered algorithm for sinkless orientation implies a $t$-round node-centered algorithm for edge-grabbing. 
\end{lemma}
\begin{proof}
We prove just the first part of the statement, the proof of the second part is similar, and doing it is a good exercise. 

We start with any node-centered algorithm $\fA$ with round complexity $t$ that solves the edge-grabbing problem. We define the edge-centered algorithm $\fA'$ of round complexity $t-1/2$ as follows. For an edge $e = uv$, the algorithm first considers all the possible extensions of the (known) ball $B(uv, t-1)$ to the ball $B(u, t)$ that is known to $\fA$ when it is run on $u$. By an extension, we mean first how the graph looks like (e.g., maybe some vertices on the boundary of $B(uv, t-1)$ turn out to be leaves), and second, what the identifiers are (they have to be consistent with $H$). We consider all valid extensions and if at least one of them leads to $\fA$ grabbing the edge $uv$ from $u$, then $\fA'$ orients the edge $uv$ from $u$ to $v$. 
After this is done, the algorithm makes an analogous reasoning for $v$; again, whenever at least one extension of $B(uv, t-1)$ to $B(v, t)$ decides to grab the edge $uv$, $\fA'$ orients it from $v$ to $u$. If no vertex ever decides to grab $uv$, $\fA'$ decides to orient it arbitrarily. This finishes the description of $\fA'$ (see \cref{fig:elimination-tree}). 

\begin{figure}
    \centering
    \includegraphics[width = .5\textwidth]{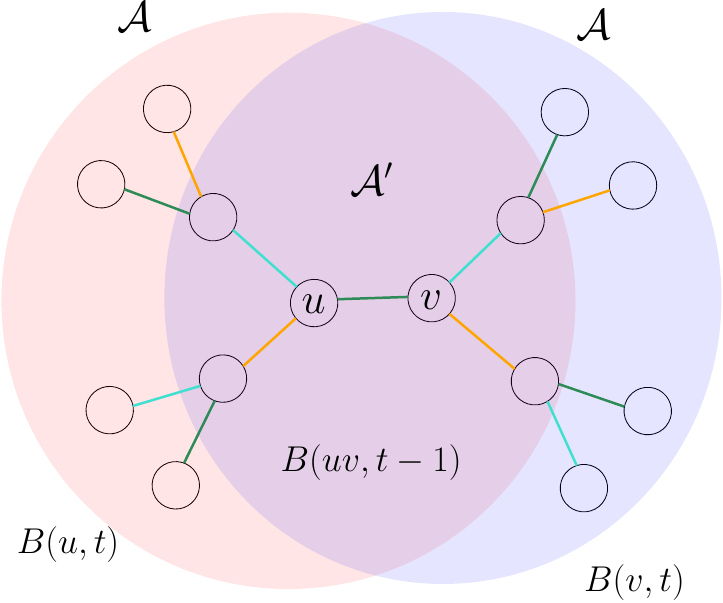}
    \caption{The picture shows the definition of the algorithm $\fA'$ for $t = 2$. Notice that our setup is a large $\Delta$-regular tree with edges colored with $\Delta$ colors. The algorithm $\fA'$ has access only to the identifiers in $B(uv, t-1)$ (the intersection of the red and the blue ball). To decide on the orientation of the edge $uv$, the algorithm considers all possible identifier extensions of $B(uv, t-1)$ to $B(u, t)$ (the red ball), and if at least one such extension makes $\fA$ run at $u$ grab the edge $uv$, $\fA'$ orients that edge towards $v$. We analogously orient this edge towards $u$ if at least one extension of $B(uv, t-1)$ to $B(v, t)$ (the blue ball) makes $\fA$ run at $v$ grab the edge $uv$. \\
    We notice that it cannot happen that $\fA'$ wants to orient the edge $uv$ in both directions: that would imply the existence of an identifier-labeling of $B(uv, t)$ on which $\fA$ is incorrect. }
    \label{fig:elimination-tree}
\end{figure}

To make sure that $\fA'$ is well-defined, we need to prove that it never happens that there is an extension of $B(uv, t-1)$ to both $B(u, t)$ and $B(v,t)$ such that $\fA$ run on $B(u, t)$ decides to grab the edge $uv$ and run on $B(v, t)$, it decides to grab the edge $vu$.
To see this, notice that putting the two extensions $B(u, t), B(v, t)$ together, the graph $B(uv,t)$ has at most $n$ nodes and its identifiers respect $H$.\footnote{This part of the argument needs to work with the ID graph instead of unique identifiers. } We observe that in case $\fA'$ is not well-defined, $\fA$ fails to solve edge-grabbing on $B(uv, t)$, a contradiction with the assumed correctness of $\fA$. 

Moreover, the algorithm $\fA'$ solves sinkless orientation: For any node $u$ of full degree $\Delta$, the original algorithm $\fA$ decided to grab a certain edge $uv$. When we run $\fA'$ on $uv$, $\fA'$ will by definition orient this edge from $u$ to $v$, thus $u$ cannot be a sink. 
\end{proof}

\paragraph{Finishing the proof}
Let us finish the proof of \cref{thm:sinkless_orientation_lb}. 
\begin{proof}[Proof of \cref{thm:sinkless_orientation_lb}]
    Consider an input graph $G$ which is a branching tree with $\eps \log n$ layers and every non-leaf vertex has degree $\Delta$. 
    Note that for small enough $\eps > 0$, the girth property of the ID graph implies that any labeling of $G$ with identifiers from $[n]$ that respects a fixed ID graph $H$ has unique identifiers. On the other hand, $G$ has $O\left( \Delta^{\eps \log n} \right)$ nodes, so the range from which the identifiers are coming is $|V(G)|^{O(1)}$. 
    Therefore, a deterministic local algorithm for sinkless orientation on $G$ implies a local algorithm for that problem on $G$ with identifiers consistent with $H$. 

    But \cref{lem:sinkless_orientation_round_elimination} shows that any $o(\log n)$-round algorithm that works relative to $H$ can be sped up to $0$ round complexity. A $0$-round algorithm $\fA_0$ is simply a function that maps an input identifier to one of $\Delta$ many colors, i.e., which edge the vertex decides to grab. We can thus think of $\fA_0$ as a coloring of $H$ with $\Delta$ many colors. But notice that for any such vertex coloring, we can consider the largest color class $i$ that has at least $n/\Delta$ colors. Then, we use the independence property of ID graphs from \cref{def:id_graph} to infer that the set of vertices of color $i$ cannot be independent in $H_i$, i.e., there is an edge $uv$ in $H$ where both $u$ and $v$, as well as the edge $uv$ are colored by the color $i$. We thus found two vertices that may be neighboring in the input graph $G$ and the algorithm $\fA_0$ decides to grab the edge connecting them from both endpoints of that edge, a contradiction with $\fA_0$ being correct.  
\end{proof}

\paragraph{Sinkless orientation as the ``simplest hard problem''}
Notice that if we formulate sinkless orientation as an instance of Lovász local lemma, the bad event probability is equal to $2^{-\Delta}$, that is, the bad event probability is exponentially small compared to the polynomial guarantee in the relaxed criterion of \cref{eq:lll_definition}. It turns out that this is the threshold where an instance of the local lemma is still ``hard''. 


\todo{fix}
\begin{restatable}[{\cite{brandt_et_al2016_LLL_lower_bound,brandt_maus_uitto2019tightLLL,brandt_grunau_rozhon2020tightLLL}}]{theorem}{sharpie}
\label{thm:sharp_threshold_onedelta_formulation}
        On one hand, there is an instance of Lovász local lemma (namely sinkless orientation) with the criterion $p \cdot 2^\Delta \le 1$ that has deterministic round complexity $\Omega(\log n)$. 
    On the other hand, any instance of Lovász local lemma with the criterion $p \cdot 2^\Delta < 1$ has deterministic round complexity $O(\log^* n)$. 
\end{restatable}

A similar threshold phenomenon holds also if we work in the variable-event graph. 

\begin{theorem}[{\cite[Theorem 3.5]{fischer_ghaffari2017sublogarithmic}, \cite[Corollary 1.8]{bernshteyn2021localcont}}]
\label{thm:sharp_threshold_twodelta_formulation}
    On one hand, there is an instance of Lovász local lemma\footnote{The instance is sinkless orientation on trees where one color class has degree $\Delta_{r.v.}$ and the other has degree $\Delta_{b.e.}$. One formulates it as an instance of the local lemma by letting each vertex of one color class grab a random outgoing edge. The proof that this variant of sinkless orientation is still hard seems to be missing in the literature. } with the criterion $p \cdot \Delta_{r.v.}^{\Delta_{b.e.}} \le 1$ has deterministic round complexity $\Omega(\log n)$. 
    On the other hand, any instance of Lovász local lemma with the criterion $p \cdot \Delta_{r.v.}^{\Delta_{b.e.}} < 1$ has deterministic round complexity $O(\log^* n)$. 
\end{theorem}

%% file: chapter2/3speedups_and_slowdowns.tex
\subsection{Speedups and Slowdowns}
\label{sec:2speedups_and_slowdowns}

The following section covers speedup and slowdown theorems which are at the heart of why we understand that there are sharp thresholds in the local complexities in \cref{thm:classification_basic}. 

The elegant idea behind speedups and slowdowns is that we can simply ``lie'' to algorithms about the size of the input graph, an idea closely related to the fooling argument we have already seen in the proof of \cref{thm:bad_id_into_good_id}. To understand this technique, it may be useful to briefly recall that in \cref{def:local_algorithm}, we defined a local algorithm as a function that takes two inputs. Firstly, it is $n$, the size of the input graph, and secondly, it is a $t(n)$-hop neighborhood of a vertex that is additionally labeled by unique identifiers or random strings. We will use the notation $\fA_n$ to denote the algorithm $\fA$ when the first input is $n$, i.e., we will view $\fA$ as a sequence of functions $\fA_1, \fA_2, \dots$ 
In this section, we will contemplate what happens if we run $\fA_n$ on a graph of size $n' \not= n$. If $n' > n$, we are ``speeding up'' $\fA$, while if $n' < n$, we are ``slowing it down''. 

\subsubsection{Slowdowns}
\label{subsec:slowdowns}
Let's first see why slowdowns may be useful. 
As a first application, let us recall that we defined deterministic and randomized round complexities in \cref{def:local_problem} by requiring unique identifiers from the range $[n^{O(1)}]$ or error probability at most $1/n^{O(1)}$. 
We will next see that for algorithms with sufficiently small round complexity, we can replace $n^{O(1)}$ by $n$ in the definition without changing its strength. For deterministic algorithms, this complements \cref{thm:bad_id_into_good_id} that shows how to replace very large identifiers with polynomially-sized ones. 
This follows by plugging in $f(n) = n^{O(1)}$ into the following slowdown theorem. 

\begin{theorem}[\citet*{chang_kopelowitz_pettie2019exp_separation}]
    \label{thm:slowdown}
    Let $f$ be any increasing function with $f(n) \ge n$, let $\Pi$ be any local problem, and let us use $t_{f(n)}(n)$ to denote the deterministic (randomized) round complexity of solving $\Pi$ if the input identifiers are from the range $[f(n)]$ (the error probability is required to be $1/f(n)$, respectively).     
    Then, $$
    t_{n}(n) \le t_{f(n)}(n) \le t_{n}(f(n)).
    $$
    
    The theorem holds for local complexities defined with respect to any subclass of graphs closed on adding isolated vertices. \footnote{Looking at the proof, it is hard to come up with a reasonable class of graphs where the theorem does \emph{not} apply. }
\end{theorem}
\begin{proof}
We will prove just the deterministic version of the theorem. Notice that $t_n(n) \le t_{f(n)}(n)$ follows directly from our assumption $f(n) \ge n$ and the definition:  any algorithm expecting identifiers from the range $[f(n)]$ certainly works if they happen to be from the smaller range $[n]$. 

Next, consider any deterministic algorithm $\fA$ that solves $\Pi$ if the identifiers are from $[n]$ in round complexity $t(n)$. Our goal is to turn it into an algorithm $\fA'$ with round complexity $t'(n)$ that works if the identifiers are from $[f(n)]$. Think of $\fA$ as a sequence $\fA_1, \fA_2, \dots$ for each $n \in \N$. We define $\fA'$ by setting for each $n$ that $\fA'_{n} := \fA_{f(n)}$. 
That is, we ``lie'' to the algorithm $\fA$, telling it that the size of the graph is larger (namely $f(n)$) than what it actually is (namely $n$). 
Since the round complexity of $\fA$ on $f(n)$-sized instances is $t(f(n))$, for the complexity of $\fA'$ on $n$-sized instances we have $t'(n) = t(f(n))$. 
    
Moreover, since $\fA_{f(n)}$ assumes that the unique identifiers are from $[f(n)]$, $\fA'_n$ also assumes that the unique identifiers from $[f(n)]$, making it a well-defined algorithm for the definition of local algorithm where identifiers are supposed to be from $[f(n)]$ on instances of size $n$. 
    
Finally, we claim that $\fA'$ is a correct algorithm. To see this, suppose that $\fA'_n$ fails to solve $\Pi$ on some graph $G$ with $n$ vertices labeled with unique identifiers from $[f(n)]$. Then, we go back to $\fA$ and run it on a graph $G'$ defined as $G$ together with $f(n) - n$ additional isolated vertices, all labeled with unique identifiers from $[f(n)]$. We notice that since $\Pi$ is a local problem, a failure in $G$ implies a failure in $G'$, and we get a contradiction with $\fA$ being correct. 
\end{proof}

As an example of a non-local problem where the range of identifiers matters, consider the leader-election problem\footnote{This is a fundamental problem in the broader area of distributed computing. The maximal independent set problem can be seen as a local variant of this problem. } where exactly one node of the input graph is to be selected. If the identifiers are from $[n]$, we can simply select the node with identifier 1. Otherwise, there is no local algorithm for it if the input graph is empty. 

A similar argument to the proof of \cref{thm:slowdown} can be used to prove another sleep-well-at-night result: any local algorithm $\fA$ with round complexity $t(n)$ solving some local problem can be turned into an algorithm $\fA'$ of round complexity $t'(n) \le t(n)$ which is a non-decreasing function of $n$. 

The randomized version of \cref{thm:slowdown} turns out to be particularly interesting. \citet*{chang_kopelowitz_pettie2019exp_separation} used it to show that \emph{any} randomized algorithm can be derandomized, if we appropriately slow down its complexity. \footnote{Their technique is similar to Adelman's theorem that proves that $\text{BPP} \subseteq \text{P}/\text{poly}$. See \cite[Theorem 7.14]{arora_barak2009computational_complexity_modern_approach}. }

\begin{theorem}[\citet*{chang_kopelowitz_pettie2019exp_separation}]
    \label{thm:exponential_derandomization}
    If a local problem $\Pi$ has randomized round complexity $t(n)$, its deterministic round complexity is $t\left(2^{O(n^2)}\right)$. 

    The theorem holds for any class of graphs closed on adding isolated vertices. 
\end{theorem}
\begin{proof}
    Let $\fA$ be a randomized local algorithm with error $1/n^{O(1)}$ and round complexity $t(n)$ solving $\Pi$. We start by applying \cref{thm:slowdown} with $f(n) = 2^{O(n^2)}$ to slow down $\fA$ to complexity $t\left( 2^{O(n^2)} \right)$ while pushing the error probability down to $2^{-\Omega(n^2)}$. We still call this new algorithm $\fA$.

    The deterministic algorithm $\fA'$ for $\Pi$ will work as follows. Let $C$ be such that the unique identifiers are coming from $[n^C]$. We will use a certain function $h$ (a hash function that we soon specify) that maps each identifier from $[n^C]$ to an infinite bit string. We define $\fA'$ as follows: It first uses $h$ to map each input identifier to an infinite bit string, and next it simulates $\fA$ using each bit string as the string of random bits. 
    
    To construct the function $h$ that makes $\fA'$ correct, we simply choose a random one and notice that on any concrete graph, $\fA'$ fails to solve $\Pi$ with probability at most $2^{-\Omega(n^2)}$ with the randomness over the choice of $h$. This is because on any concrete graph, $\fA'$ with random $h$ is equivalent to running $\fA$. But notice that the number of distinct graphs on $n$ vertices labeled with polynomial identifiers is $2^{O(n^2)} \cdot (n^C)^n = 2^{O(n^2)}$.  
    Hence, choosing the failure probability small enough, we can union-bound over all distinct labeled graphs and still get that random $h$ works, with positive probability. We thus conclude that for some $h$, $\fA'$ is a valid deterministic algorithm for $\Pi$, as needed. 
\end{proof}

As a nice corollary, we are getting a very precise understanding of what is happening for the randomized round complexity of Lovász local lemma: On one hand, we have seen a randomized algorithm with round complexity $\poly\log\log(n)$ (\cref{thm:fischer_ghaffari}) that used a deterministic algorithm with round complexity $\poly\log(n)$ (\cref{cor:deterministic_lll}) as a subroutine. But we can now see that \emph{any} $\poly\log\log(n)$ round randomized algorithm would imply a $\poly\log(n)$-round deterministic one via \cref{thm:exponential_derandomization}. A similar observation applies not just for the local lemma, but for many other problems with randomized algorithms based on shattering (see \cref{sec:3concrete_problems}). This explains why in local complexity deterministic algorithms are a big deal -- not only they are used as subroutines of randomized algorithms, but we in fact understand that to make progress in the randomized world, progress in the deterministic world is often necessary. 

\subsubsection{Speedups}
\label{subsec:speedups}

While slowdowns are very useful, the real fun starts when we try to speed up algorithms by lying to them that the size of the graph is \emph{smaller} than it actually is. In fact, what can possibly go wrong if we tell an algorithm that the size of the graph is constant, i.e., we claim that $n = O(1)$? As it turns out, not much! One potential problem is that the algorithm may find out that its $t(n)$-hop neighborhood contains more than $n$ nodes and ``crash''. The other problem is that the algorithm now requires very small identifiers (if we started with a deterministic algorithm) or that the failure probability of the algorithm increased to constant (if we started with a randomized one). This is where our story connects with our discussions of coloring and Lovász local lemma.  

\begin{theorem}[\citet*{chang_pettie2019time_hierarchy_trees_rand_speedup,chang_kopelowitz_pettie2019exp_separation}]
\label{thm:deterministic_speedup}
Let $\fA$ be a local algorithm for a local problem $\Pi$ with round complexity $t(n) = o(\log n)$.  Then there is a local algorithm $\fA'$ solving $\Pi$ such that
\begin{enumerate}
    \item If $\fA$ was deterministic, so is $\fA'$ and its round complexity is $O(\log^* n)$. 
    \item If $\fA$ was randomized, so is $\fA'$ and its round complexity is the same as the randomized complexity of solving relaxed Lovász local lemma, i.e., it is $\tilde{O}\left( \log^4\log(n) \right)$.  
\end{enumerate}

The theorem holds for the class of bounded degree graphs and any of its subclasses closed on taking subgraphs and adding isolated vertices. More generally, for round complexities defined relative to a subclass of bounded degree graphs of restricted growth, the complexity $t(n)$ needs to be such that we always have $|B(u, t(n))| = o(n)$ for any node $u$. 

\end{theorem}
\begin{proof}
    Let $\Pi$ be a local problem with a checkability radius of $r$ and $\fA$ be a deterministic (randomized) local algorithm for $\Pi$. 
    We will assume that the identifiers are coming from a range $[n^C]$ (alternatively, the failure probability is $1/n^C$) for sufficiently large $C$.  
    We will view $\fA$ as a sequence $\fA_1, \fA_2, \dots$, and we start by choosing $n_0$ to be a large enough constant so that 
\begin{align}
\label{eq:speedup}
    \Delta^{0} + \Delta^{1} + \dots + \Delta^{t(n_0) + r} \le n_0. 
\end{align}
    We notice that such an $n_0$ exists by the requirement $t(n) = o(\log n)$.     
    Our algorithm $\fA'$ will simulate $\fA_{n_0}$ for \emph{any} input $n$ (for $n < n_0$ we define $\fA'_n = \fA_n$). 

    \begin{figure}
        \centering
        \includegraphics[width = \textwidth]{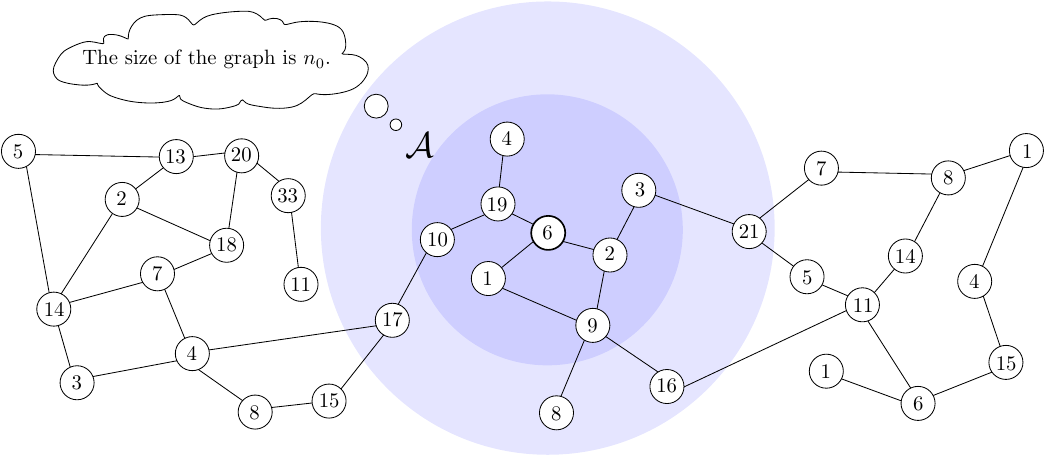}
        \caption{To speed up an algorithm $\fA$, we tell it that the size of the graph is some constant $n_0$, whatever the true size $n$ is. In this case, the round complexity of the algorithm becomes $2$ (see the smaller ball in the picture) and we assume that the checkability radius of the local problem solved by $\fA$ is $1$. \\
        The algorithm $\fA$ assumes that the identifiers are unique numbers of size up to $n_0^{O(1)}$. 
        Since $n \gg n_0$, we cannot supply such identifiers. However, we notice that if $\fA$ is supplied with $n_0^{O(1)}$-sized identifiers that are unique only in every ball of radius $3$ (see the larger ball in the picture), the algorithm still has to work. This is because a failure of $\fA$ around some node $u$ implies a failure of $\fA$ on the graph $B(u, 3)$ which has size at most $n_0$ and is labeled with unique identifiers. }
        \label{fig:speedup}
    \end{figure}

    Let us first handle the case of deterministic algorithms (see \cref{fig:speedup}). There, $\fA'$ needs to supply identifiers to its simulation of $\fA_{n_0}$. This is done as follows. We compute a $(\Delta(G')+1)$-coloring of the power graph $G' = G^{2(t(n_0)+r)}$ such that the number of colors used is $\Delta(G')+1 = \Delta^{O(t(n_0) + r)} \le n_0^{C}$ where we used that $C$ is large enough. We use those colors as identifiers for our simulation of $\fA_{n_0}$. This finishes the description of $\fA'$. 
    
    On one hand, the round complexity of the overall algorithm is $O(\log^* n)$ on bounded-degree graphs. On the other hand, suppose that $\fA'$ fails to solve the problem $\Pi$ at some node $u$. We notice that \cref{eq:speedup} yields that the $(t(n_0)+r)$-hop neighborhood of $u$ contains at most $n_0$ nodes. Moreover, this neighborhood is colored by $\fA'$ with unique colors from the range $[n_0^{C}]$. Recall that $\fA_{n_0}$ is defined on neighborhoods with at most $n_0$ nodes labeled with identifiers from $[n_0^{C}]$. Thus the algorithm $\fA'$ is well-defined. Additionally, a failure of $\fA'$ at $u$ would imply a failure of $\fA$ on a neighborhood of size at most $n_0$ labeled with some unique identifiers from $[n_0^{C}]$. Adding isolated vertices to that neighborhood to make its size $n_0$, we reach a contradiction with $\fA_{n_0}$ being correct on $n_0$-sized instances. 

    \vspace{1em}
    We handle the case of randomized algorithms similarly: We again try to simulate $\fA_{n_0}$ for large enough $n_0$. We notice that we do not need to supply any identifiers in the simulation. However, if we simply simulate $\fA$, we fail to solve $\Pi$ at any fixed vertex $u$ with probability $1/n_0^{C}$ as this is the failure probability of $\fA_{n_0}$. Thus, we instead formulate the process of simulating $\fA$ as \emph{an instance of Lovász local lemma} on a new variable-event graph $G'$ as follows. Any original vertex $u$ of the input graph $G$ will be thought of as two new vertices in $G'$. 
    One vertex, $u_{r.v}$, is a random-variable vertex, and it corresponds to the random string at $u$. The other vertex, $u_{b.e.}$, corresponds to the event that if we run $\fA_{n_0}$ at $u$ and use the random strings sampled from the random-variable vertices $v_{r.v.} \in B(u_{b.e.}, t(n_0))$, $\fA_{n_0}$ fails to solve $\Pi$ at $u$. 

    We notice that the event at $u_{b.e.}$ depends only on the $(t(n_0) + r)$-hop neighborhood of $u$ in $G$. Thus the degree of this local-lemma instance (in the sense of its dependency-graph formulation) is at most $\Delta^0 + \Delta^1 + \dots + \Delta^{t(n_0) + r}$ which is in turn at most $n_0$ by \cref{eq:speedup}. On the other hand, the probability of each bad event is at most $1/n_0^C$. 
    We can thus formulate the simulation of $\fA_{n_0}$ as an instance of a $C$-relaxed local lemma. For large enough $C$, we can then apply \cref{thm:fischer_ghaffari} to solve that instance. 

\vspace{1em}
    Finally, we notice that both in the deterministic and the randomize case, we can generalize the requirement $t(n) = o(\log n)$ to requiring that $t$ satisfies for each $u$ that $|B(u, t(n))| = o(n)$. 
\end{proof}
We note that in the statement of \cref{thm:deterministic_speedup}, we do not literally require $t(n) = o(\log n)$ (or $|B(u, t(n))| = o(n)$), it should just be that for some large enough $n_0$, we have $t(n_0) \le \eps \log n_0$, where $\eps$ is some small constant which is a function of $\Delta, r$, and $C$. 

\paragraph{Additional intuition yielded by the proof}
The proof of \cref{thm:deterministic_speedup} tells us a bit more: Basically, it allows us to think of local problems with deterministic round complexity $o(\log n)$ or, equivalently, $O(\log^* n)$, as ``those problems that can be solved in a constant number of rounds, after we compute a distance coloring''. Moreover, the problems with randomized round complexity $o(\log n)$ or, equivalently $\poly\log\log n$, can be thought of as ``those problems that we can view as an instance of the local lemma''. 

\tbox{Problems with deterministic round complexity $o(\log n)$ can be, in fact, solved in constantly many rounds, after a suitable coloring used as ``fake identifiers'' is computed. \\ Problems with randomized round complexity $o(\log n)$ can be formulated as instances of polynomially relaxed Lovász local lemma. }

\paragraph{Speedup below $\log^* n$}
We will next sketch how one can speed up algorithms with round complexity $o(\log \log^* n)$ to complexity $O(1)$. Why the weird complexity $o(\log\log^* n)$? The following speedup argument relies more on the volume than on the radius\footnote{Compare with \cref{thm:classification_volume} classifying volume complexities where we have $o(\log^* n)$ speedup instead. }, so the importance of radius $o(\log\log^* n)$ is that the number of nodes the algorithm sees is $\Delta^{o(\log\log^* n)} = o(\log^* n)$. 
The $\log^* n$ volume is then the right threshold for which we can argue that any deterministic algorithm can be turned into a drastically simpler \emph{order-invariant} algorithm. \footnote{While the name order-invariant suggests that the algorithm may not care about the order of identifiers in some way, it is the other way around, the algorithm cares \emph{only} about the order of the identifiers in its neighborhood. That is, the algorithm is \emph{comparison-based}. } 

An order-invariant algorithm is a deterministic local algorithm with an additional restriction: it does not have direct access to the identifiers written on the nodes; it only knows their \emph{order}, meaning that it can only compare their relative size. 
A bit more formally, while an input to a classical local algorithm is a graph $G$ together with a map $f: V(G) \rightarrow [n^{O(1)}]$ that maps vertices to unique identifiers, an input to an order-invariant algorithm is $G$ together with a function $f': V(G) \times V(G) \rightarrow \{<, > \}$ that is consistent with a total linear order on $V(G)$ and that can tell the algorithm for each pair of nodes which one is larger. Of course, $t(n)$-round algorithm run at $u$ only has access to $f'$ restricted to $B(u, t(n)) \times B(u, t(n))$.

For example, in the setup of \cref{thm:coloring_lower_bound} where we worked with an oriented path with increasing labels, an order-invariant algorithm would see the same input order at every vertex. 
This already shows that order-invariant algorithms struggle to solve interesting local problems on graphs that are paths.

\begin{theorem}[\citet*{naorstockmeyer,chang_kopelowitz_pettie2019exp_separation}]
    \label{thm:speedup_logstar}
    Let $\fA$ be a deterministic local algorithm for a local problem $\Pi$ with round complexity $o(\log \log^* n)$. Then, there is a deterministic local algorithm $\fA'$ solving $\Pi$ with local complexity $O(1)$. 
    
    The theorem holds for any subclass of bounded-degree graphs closed on taking subgraphs and adding isolated nodes. 
\end{theorem}
\begin{proof}[Proof Sketch]
    
    We will sketch how any deterministic local algorithm $\fA$ with round complexity $t(n) = o(\log \log^* n)$ can be turned into an order-invariant algorithm $\fA'$, using the hypergraph Ramsey theorem\footnote{Recall that in hypergraph Ramsey theorem we have a big set $B$, we color its small subsets $S$ and we want to find a medium-sized set $M$ such that all its small subsets $S \subseteq M$ have the same color. }. Fix an algorithm $\fA$ and consider the set $B = [n^{O(1)}]$ of identifiers that $\fA$ expects on input. Next, consider any subset $S \subseteq B$ of size $(2\Delta)^{t(n)} = o(\log^* n)$. Here, we have chosen the function $(2\Delta)^{t(n)}$ because it upper-bounds the maximum number of nodes in any $t(n)$-hop neighborhood that $\fA$ can see as its input. 
    
    Next, let us define the color of $S$. To do so, first consider the set $\fH$ of all possible $t(n)$-hop neighborhoods. Moreover, for each (unlabeled) neighborhood $H \in \fH$, consider all of at most $|S|!$ ways in which the identifiers from $S$ can be assigned to the vertices of $H$. For each labeled neighborhood, we record the output of $\fA$ when it is run on it, and the color of $S$ is defined to be the complete list of all such outputs. 
    
    We note that the number of possible colors of a set $S \subseteq B$ is relatively small, and in particular it can be upper-bounded by $2^{2^{O(\log^* n)^2}}$. 
    Recalling that we use $r$ to denote the checkability radius of $\Pi$, we aim to find a medium-sized set $M \subseteq B$ of size $(2\Delta)^{t(n) + r} = o(\log^* n)$ such that all subsets $S \subseteq M$ have the same color. 
    The known bounds for hypergraph Ramsey numbers \cite[§1, Theorem 2]{graham1991ramsey} allow us to conclude that $B$ is large enough to always contain such a set $M$.

    We notice that the algorithm $\fA$ behaves in the order-invariant way if the input identifiers are coming from the set $M$. This allows us to define an order-invariant $t(n)$-round algorithm $\fA'$ for $\Pi$ as follows. 
    Given an input order on the vertices of $B(u, t(n))$, the algorithm $\fA'$ chooses an arbitrary subset $S \subseteq M$ and maps the identifiers in $S$ to vertices of $B(u, t(n))$ in the unique order-preserving way. Concretely, whenever $v_1 < v_2$ according to the input order, we also have $ID(v_1) < ID(v_2)$. Our algorithm then simply outputs what $\fA$ would output on this input. This finishes the description of $\fA'$. 

    To see that we are on the right track, we first notice that it does not matter which subset $S \subseteq M$ we choose above; since all subsets $S$ have the same color, they all lead to the same output. 

    Next, to see that $\fA'$ is valid, we assume for contradiction that it fails at a vertex $u$ in some graph $G$ equipped with some input order. We consider the ball $B(u, t(n) + r)$ and the restriction of the input order to this ball.
    Consider the following thought experiment: In the ball $B(u, t(n)+r)$, we replace the input order with identifiers from $M$. We do it in the unique way that preserves the original order; that is, if $v_1 < v_2$ in the original order, then $ID(v_1) < ID(v_2)$. We observe that for any vertex $v \in B(u, r)$, the algorithm $\fA'$ run on $v$ with access to the original input order returns the same as $\fA$ run on $v$ with access to the identifiers from $M$. This is because $\fA'$ returns the same answer, whatever identifiers from $M$ are chosen to replace the input order. 
    
    But this means that the original algorithm $\fA$ also fails at the node $u$ for the specific choice of identifiers from $M$ above. After adding isolated nodes to the ball $B(u, t(n))$ equipped with these identifiers, we find a graph on $n$ vertices where $\fA$ fails and we thus reach a contradiction. 

    Finally, let us contemplate the proof of \cref{thm:deterministic_speedup} again. The only reason why we could speedup $o(\log n)$-round deterministic algorithm only to $O(\log^* n)$ instead of $O(1)$ was the necessity to compute a coloring that served as ``fake'' identifiers supplied to the simulation of an algorithm $\fA_{n_0}$. But for order-invariant algorithms, this is not necessary. We can simply give $\fA_{n_0}$ exactly the same order as the order that appears on input; thus order-invariant algorithms can be sped up to complexity $O(1)$ and this finishes the proof. 
\end{proof}

\paragraph{Tower-sized identifiers}
We can wonder what would happen if we wanted to make the above proof of \cref{thm:speedup_logstar} work for \emph{all} deterministic local algorithms, not just those that encounter only $o(\log^* n)$ vertices. 
This would force us to change the small-set size from $o(\log^* n)$ to $n$ in the proof. Since the hypergraph Ramsey numbers grow roughly as a tower function of the small-set size, this means that to make the proof work, we would have to start with unique identifiers from the range which is contains numbers as large as the tower function of $n$. The rest of the proof then goes through and allows us to speed up $o(\log n)$-round algorithms to $O(1)$-round order-invariant algorithms. 

This observation nicely complements \cref{thm:bad_id_into_good_id}: While it is true that exponentially-, doubly-exponentially-, etc. sized identifiers have the same power as polynomially-sized ones, there is a transition somewhere around tower-function-sized identifiers. From then on, deterministic algorithms of $o(\log n)$ complexity with such huge identifiers can be turned into order-invariant ones. In other words, the identifiers are so large that no $o(\log n)$-round algorithm can extract any meaningful information from them other than their relative order. 

%% file: chapter2/4classification.tex
\subsection{Classification of Local Problems}
\label{sec:3concrete_graphs}

We can now put all the pieces together and prove \cref{thm:classification_basic}. Let's restate it here for convenience. 

\classification*
\begin{proof}
    Let $\Pi$ be any local problem of randomized local complexity $t(n) = o(\log n)$. We can use \cref{thm:deterministic_speedup} to formulate $\Pi$ as an instance of Lovász local lemma. This instance is then solved either with the deterministic $\tilde{O}(\log^4 n)$-round algorithm from  \cref{cor:deterministic_lll}, or the randomized $\tilde{O}(\log^4 \log n)$-round algorithm from \cref{thm:fischer_ghaffari}. 

    Next, let us assume that the randomized round complexity of $\Pi$ is even $o(\log\log n)$. Then, we can use the derandomization of \cref{thm:exponential_derandomization} to conclude that the deterministic round complexity is $o(\log n)$. This allows us to use the speedup theorem of \cref{thm:speedup_logstar} to conclude that the deterministic round complexity of $\Pi$ is $O(\log^* n)$. Moreover, we notice that for any increasing function $t(n) = O(\log^* n)$, we have $t(2^{O(n^2)}) = O(t(n))$; this means that the derandomization from \cref{thm:exponential_derandomization} tells us that randomized and deterministic round complexities of $\Pi$ are the same. 

    Finally, assume that the randomized local complexity of $\Pi$ is even $o(\log\log^* n)$. Then, we can use \cref{thm:speedup_logstar} to find an order-invariant constant-round local algorithm for $\Pi$ via \cref{thm:speedup_logstar}. 
\end{proof}

\tbox{
Local problems with complexities below $o(\log n)$ are of three types: 
\begin{enumerate}
    \item Those that can be solved in constant rounds by order-invariant algorithms, 
    \item those that are roughly as hard as the basic symmetry-breaking problems such as coloring or maximal independent set, 
    \item those that can be viewed as instances of Lovász local lemma. 
\end{enumerate}
}

\paragraph{Improvements to \cref{thm:classification_basic}}
Can we improve \cref{thm:classification_basic}? The conjecture of \citet*{chang_pettie2019time_hierarchy_trees_rand_speedup} (\cref{conj:chang_pettie}) suggests that in the local lemma regime, we may get rid of the polynomial gap between the lower and upper bounds. According to our current knowledge, there could be a local problem with deterministic and randomized round complexities $t_d(n), t_r(n)$ for any $t_d(n) = \Omega(\log n)$, $t_d(n) = \tilde{O}(\log^4 n)$, $t_r(n) = \Omega(\log\log n)$, $t_r(n) = \tilde{O}(\log^4\log n)$, and $t_d(n) = O(t_r(2^n))$, where the last constraint follows from \cref{thm:slowdown}. 

On the other hand, it is known \cite{balliu2018new_classes-loglog*-log*} that the (extremely narrow) regime between $\Omega(\log\log^* n)$ and $O(\log^* n)$ is densely populated by certain (artificial) local problems. 

\subsubsection{Sequential local complexities}
Let us now discuss how sequential local complexities fit into the picture painted by \cref{thm:classification_basic}. 
We will first observe that the first two classes from our classification in \cref{thm:classification_basic}, i.e., $O(1)$ and $O(\log^* n)$ round complexities, are equal to the class of problems solvable with the sequential local complexity is $O(1)$. 

\begin{theorem}
\label{thm:logstar_distributed_equal_O1_sequential}
    On bounded degree graphs, local problems with round complexity $O(\log^* n)$ are exactly those whose sequential local complexity is $O(1)$. 
\end{theorem}
\begin{proof}[Proof sketch]
On one hand, consider any local problem with constant sequential local complexity. We can use \cref{thm:sequential_vs_distributed_coloring} to simulate it with a distributed local algorithm in $O(\log^* n)$ rounds.     

On the other hand, consider any local problem $\Pi$ with local checkability $r$ solvable in $O(\log^* n)$ round complexity. Recall that the proof of \cref{thm:deterministic_speedup} says that $\Pi$ can in fact be solved with a local algorithm $\fA$ of round complexity $t = O(1)$ on any graph $G$ if we are provided a proper coloring of a suitable power graph $G^{2(t+r)}$.
Recall that $(\Delta+1)$-coloring has sequential local complexity $O(1)$, meaning that we can construct a sequential local algorithm of constant complexity for $\Pi$. 
\end{proof}

Next, we will show how our distributed speedup theorems imply speedups in the sequential world. 
\begin{theorem}
    \label{thm:sequential_speedup}
    Let $\fA$ be a sequential local algorithm solving a local problem $\Pi$ on bounded degree graphs. Assume that either 
    \begin{enumerate}
        \item $\fA$ is deterministic and its sequential local complexity is $o(\log\log n)$,
        \item or $\fA$ is randomized and its sequential local complexity is $o(\log\log\log n)$. 
    \end{enumerate}
    Then, there is a distributed local algorithm $\fA'$ that solves $\Pi$ in $O(\log^* n)$ round complexity. 
\end{theorem}
\begin{proof}


    If $\fA$ is a deterministic sequential algorithm of local complexity $t(n) = o(\log\log n)$, we can simulate it with a distributed local algorithm via \cref{thm:sequential_vs_distributed_coloring}. We get a deterministic algorithm with local complexity $\Delta^{O(t(n))} + O(t(n) \log^* n) = \log^{o(1)} n = o(\log n)$. Such an algorithm is then sped up to $O(\log^* n)$ complexity using \cref{thm:deterministic_speedup}. 

    Similarly, if $\fA$ is randomized sequential algorithm of local complexity $o(\log\log\log n)$, \cref{thm:sequential_vs_distributed_coloring} implies that we get a randomized distributed local algorithm with $\log^{o(1)} \log n = o(\log\log n)$ round complexity. Such an algorithm can then be derandomized by \cref{thm:exponential_derandomization} into a deterministic algorithm of local complexity $o(\log n)$ which is in turn sped up to $O(\log^* n)$ using \cref{thm:deterministic_speedup}.     
\end{proof}

Finally, we sketch how one can construct very fast sequential local algorithms for instances of the local lemma. 
\begin{theorem}
    \label{thm:sequential_lll}
    Let $\Pi$ be a $C$-relaxed Lovász local lemma for sufficiently large $C$. Then, its deterministic sequential local complexity is $\tilde{O}(\log^4 \log n)$ and its randomized sequential local complexity is $\tilde{O}(\log^4\log\log n)$\footnote{It is an interesting exercise to open up the proof that the randomized sequential local complexity of local lemma is $\Theta(\poly\log\log\log n)$ and try to think of theorems we have seen so far that do \emph{not} go into that proof.}, on bounded degree graphs. 
\end{theorem}
\begin{proof}[Proof Sketch]
    To prove the first part of the theorem, we consider the randomized distributed local algorithm for local lemma from \cref{thm:fischer_ghaffari}. We turn it into a deterministic sequential local algorithm of the same complexity by the derandomization result of \cref{thm:derandomization_sequential}.  

    To prove the second part of the theorem, we notice that our randomized distributed local algorithm from \cref{thm:fischer_ghaffari} can certainly be implemented in the model of randomized sequential local algorithms. 
    Recall that the algorithm reduces an instance of the local lemma to instances on graphs of size $O(\log n)$, as proven in \cref{thm:fischer_ghaffari_first_phase}. Then, the best deterministic algorithm is applied to those instances to finish the job. While \cref{thm:fischer_ghaffari_first_phase} is stated as reducing from $C$-relaxed local lemma with some large $C$ to local lemma with condition $p \le 1/(3\Delta)$, we can generalize it to show that starting with a very large $C'$, we can reduce an instance of $C'$-relaxed local lemma to $C$-relaxed local lemma for $C$ that is still arbitrarily large.  
    
    This allows us to run the deterministic sequential algorithm with complexity $\tilde{O}(\log^4\log n)$ constructed in the first part of the proof instead of using the deterministic $O(\poly\log n)$-round algorithm from \cref{cor:deterministic_lll}. 
    The final randomized sequential local complexity is thus $\tilde{O}(\log^4\log\log n)$. 
\end{proof}
We note that for problems like sinkless orientation for which an $O(\log\log n)$-round randomized algorithm is known \cite{ghaffari_harris_kuhn2018derandomizing}, the proofs of above \cref{thm:sequential_speedup,thm:sequential_lll} imply tight bounds of $\Theta(\log\log n)$ and $\Theta(\log\log\log n)$ deterministic and randomized sequential local complexity on bounded degree graphs. In the case of sinkless orientation, the same bounds are also known even on unbounded degree graphs \cite{ghaffari_harris_kuhn2018derandomizing}. If the conjecture of \citet*{chang_pettie2019time_hierarchy_trees_rand_speedup} (\cref{conj:chang_pettie}) is true, then there is no gap between \cref{thm:sequential_speedup,thm:sequential_lll}. 

We also note that \cref{thm:sequential_speedup,thm:sequential_lll} show that while one can turn randomized distributed local algorithms into sequential deterministic ones via \cref{thm:derandomization_sequential}, it is \emph{not} possible to turn randomized sequential local algorithms into deterministic sequential local ones, without loss in their local complexity. 

Putting everything together, we get the following refined classification theorem that also includes sequential local complexities. 

\begin{theorem}[Refined classification of local problems]
    \label{thm:classification_extended}
    Let us fix any $\Delta$ and any class of graphs $\fG$ of degree at most $\Delta$ that is closed under taking subgraphs and adding isolated nodes. Let $t(n)$ be a function such that for any node $u$ in some graph in $\fG$ we have $|B(u, t(n))| = o(n)$. 
    
    Then, any local problem with randomized round complexity at most $t(n)$ has one of the following three complexities. 
    \begin{enumerate}
        \item \emph{Order-invariant regime:} The problem has $O(1)$ deterministic (or randomized) round complexity and sequential local complexity.   
        \item \emph{Symmetry-breaking regime:} The deterministic and randomized round complexity of the problem lies between $\Omega(\log\log^* n)$ and $O(\log^* n)$ (both the deterministic and the randomized complexity is the same function). Moreover, the sequential local complexity is $O(1)$. 
        \item \emph{Lovász-local-lemma regime:} The problem has the following complexities: 
        \begin{enumerate}
            \item deterministic round complexity is between $\Omega(\log n)$ and $\tilde{O}(\log^4 n)$,
            \item randomized round complexity and deterministic sequential local complexity is between $\Omega(\log\log n)$ and $\tilde{O}(\log^4\log n)$, 
            \item randomized sequential local complexity is between $\Omega(\log$ $\log\log n)$ and $\tilde{O}(\log^4\log \log n)$. 
        \end{enumerate}
    \end{enumerate}
    Moreover, if any graph from the class satisfies that for any node $u$, we have $|B(u, r)| = 2^{O(r^{0.249})}$, then there are no local problems in the third class of problems. 
\end{theorem}

The last part of the theorem follows by applying \cref{thm:deterministic_speedup}: We know that every local problem in the local-lemma regime can be solved with a deterministic local algorithm of round complexity $\tilde{O}(\log^4 n)$; since $2^{\left( \tilde{O}(\log^4 n) \right)^{0.249}} = o(n)$, the speedup to $O(\log^* n)$ complexity can be applied. If the conjecture of \citet*{chang_pettie2019time_hierarchy_trees_rand_speedup} (\cref{conj:chang_pettie}) is true, we can get rid of the polynomial slack in the local-lemma regime complexities and the whole local-lemma class of problems is then empty for the class of subexponential growth graphs (see \cref{subsec:classification}). 

\begin{problem}
    Can local problem solvable with deterministic sequential local complexity $o(\log n)$ be sped up to $\poly\log\log n$ deterministic sequential complexity (i.e., to the local-lemma regime)? What about problems with randomized sequential local complexity $o(\log n)$? 
\end{problem}

%% file: chapter2/5listing_classes.tex
\subsubsection{Classification of Local Problems for Concrete Graph Classes}
\label{subsec:classification}

This section compiles results from a long line of work trying to extend the classification of local problems on bounded-degree graphs, i.e., \cref{thm:classification_basic}, to more specific graph classes. For some simple graph classes, we now have an almost complete understanding of the local complexity classes that go even beyond the $\log n$ local complexity threshold. 

\paragraph{Locally checkable labeling problems}
The theorems below are proven for the so-called \emph{locally checkable labeling problems}. 
These are local problems that additionally allow one of the finitely many input colors on each node of the input graph. For example, a list coloring is an example of a locally checkable labeling problem. 

\begin{definition}
    \label{def:lcl}
    A \emph{locally checkable labeling problem} (LCL) $\Pi$ is formally a quintuple $(S_{in}, S_{out}, r, \Delta, \fP)$. 
    Here, $S_{in}$ is the finite set of \emph{input labels}, $S_{out}$ is the finite set of \emph{output labels}, and $\fP$ is the set of allowed neighborhoods: Each neighborhood has maximum degree $\Delta$, radius $r$, and each node is labeled with a label from $S_{in}$ and $S_{out}$. 
    Given an input graph of degree at most $\Delta$ and any input coloring of its nodes with $S_{in}$, the task is to find an output coloring by colors from $S_{out}$ so that the $r$-hop neighborhood of each node is in $\fP$.  
\end{definition}

We did not have to distinguish between our definition of a local problem from \cref{def:local_problem} and the locally checkable labelings since, in the general graph classes like the class of all graphs or all bounded-degree graphs, it is straightforward to encode input labels as a part of the input graph. We thus chose the simpler definition that also allowed us to talk about local problems even outside of the setup of bounded-degree graphs. 
In the following theorems, especially for very restrictive classes like paths or grids, the difference between the two definitions matters, since the possibility of allowing inputs may substantially enlarge the number of possible local problems. 

In the following theorems, we will always implicitly assume that they are for \emph{solvable} locally checkable labelings, i.e., for problems where the solution is always guaranteed to exist. Otherwise, the problem does not have a well-defined local complexity class.  

\paragraph{Computational aspect}
We note that an additional dimension for all classification results below is the computational complexity of the classification problem, where the input is a description of locally checkable labeling, and the output is which class it belongs to. On the one hand, it is known that the classification problem is decidable on paths, albeit PSPACE-hard \cite{balliu2019LCLs_on_paths_decidable}. On the other hand, the classification problem is undecidable even for grids \cite{brandt_grids}. The decidability on trees is open; see \cref{prob:tree_decide}. We will not discuss this dimension in the following theorem statements for brevity. 

\paragraph{Additional remarks}
We list the concrete classification theorems below. Unless stated otherwise, the deterministic and the randomized round complexity of the problem is always the same. We are always citing papers proving results specific to the given graph class, for example, general speedup and slowdown theorems of \citet*{chang_pettie2019time_hierarchy_trees_rand_speedup,chang_kopelowitz_pettie2019exp_separation} are as a rule of thumb always an important part of the proof. Also, not all the papers we cite prove a part of the given theorem. Some of them only provide building blocks, analyze the computational complexity of the classification, etc. 

\paragraph{List of known classifications}
Let us now list known classification theorems on concrete graph classes. 

\begin{theorem}[Classification of local problems on paths \cite{brandt_grids,balliu2019LCLs_on_paths_decidable}]
\label{thm:classification_paths} 
Any solvable locally checkable labeling problem on oriented or unoriented paths with inputs has one of the following round complexities:
\begin{enumerate}
    \item $O(1)$,
    \item $\Theta(\log^* n)$,
    \item $\Theta(n)$. 
\end{enumerate}
\end{theorem}

\begin{restatable}[Classification of local problems on grids \cite{naorstockmeyer,brandt_grids,chang_kopelowitz_pettie2019exp_separation,grunau_rozhon_brandt2022speedups_below_log_star}]{theorem}{gridsclassify}
\label{thm:classification_grids} 
Any solvable locally checkable labeling problem on $d$-dimensional oriented grids\footnote{The edges of the input grid are consistently oriented. } has one of the following local complexities:
\begin{enumerate}
    \item $O(1)$,
    \item $\Theta(\log^* n)$,
    \item $\Theta(n^{1/d})$. 
\end{enumerate}
\end{restatable}
We note that on unoriented grids, the full classification is not known \cite{grunau_rozhon_brandt2022speedups_below_log_star}. 

\begin{theorem}[Classification of local problems on bounded-degree rooted regular trees \cite{balliu2021rooted_trees,balliu_brandt_chang_olivetti_studeny_suomela2022efficient_regular_trees}]
\label{thm:classification_rooted_trees} 
Any solvable locally checkable labeling problem on bounded-degree rooted regular trees\footnote{Edges are oriented such that every vertex has at most one ingoing edge. } has one of the following round complexities:
\begin{enumerate}
    \item $O(1)$,
    \item $\Theta(\log^* n)$,
    \item $\Theta(\log n)$,
    \item $\Theta(n^{1/k})$ for some $k \in \N$. 
\end{enumerate}
\end{theorem}

\begin{restatable}[Classification of local problems on bounded-degree trees \cite{MillerReif1989rake_and_compress,chang_pettie2019time_hierarchy_trees_rand_speedup,balliu2018new_classes-loglog*-log*,chang2020n1k_speedups,balliu2020almost_global_problems,grunau_rozhon_brandt2022speedups_below_log_star,balliu_brandt_chang_olivetti_studeny_suomela2022efficient_regular_trees,balliu2020classification_binary}]{theorem}{treesclassify}
\label{thm:classification_trees}    
Any solvable locally checkable labeling problem on trees has one of the following local complexities:
\begin{enumerate}
    \item $O(1)$,
    \item $\Theta(\log^* n)$,
    \item $\Theta(\log\log n)$ randomized, $\Theta(\log n)$ deterministic,
    \item $\Theta(\log n)$ (both randomized and deterministic),
    \item $\Theta(n^{1/k})$ for some $k \in \N$. 
\end{enumerate}
\end{restatable}

\begin{problem}
\label{prob:tree_decide}
    Given an input locally checkable labeling problem on bounded-degree trees, is it decidable which class in \cref{thm:classification_trees} it belongs to? 
\end{problem}

\paragraph{More Classifications}
It is an intriguing question of how far these classifications can be extended, see e.g. the recent work extending the theory to minor-closed classes \cite{chang2023classification_beyond_trees} or extending the theory to unbounded-degree graphs \cite{lievonen2023classification_high_degree}. 

Another interesting class is the class of subexponential growth graphs, i.e., the class of bounded degree graphs where for every $\eps > 0$ we can find $r$ such that $|B(u, r)| \le (1+\eps)^r$ \footnote{Formally, this is a family of classes, with one class for each possible dependency $r = r(\eps)$. }. 
In this class of graphs, we can apply \cref{thm:deterministic_speedup} to speed up deterministic algorithms of complexity $O(\log n)$ (instead of $o(\log n)$) to $O(\log^* n)$. If the conjecture of \cite{chang_pettie2019time_hierarchy_trees_rand_speedup} (\cref{conj:chang_pettie}) is true, this has an interesting corollary: the local lemma regime from the classification theorem of \cref{thm:classification_basic} would then not be present in this class of graphs and there would only two classes of local problems there (see \cref{thm:classification_extended}). 
This leads to the following special case of \cref{conj:chang_pettie}:
\begin{problem}
    \label{prob:subexponential_lll}
    Is there a $C$ such that all instances of the $C$-relaxed Lovász local lemma can be solved with round complexity $\Theta(\log^* n)$ on any class of graphs of subexponential growth? 
\end{problem}
Some implications of a positive answer to this problem are known \cite{csoka_grabowski_mathe_pikhurko_tyros2022borel_lll,bernshteyn2023borel_vizing}. 

%% file: chapter4/1applications_to_parallel.tex
\section{Applications}
\label{chap:3_concrete_results}

While the previous two sections, \cref{chap:1_local_complexity_fundamentals,chap:2_below_logn}, discussed the general theory of local algorithms, the following section considers various popular models of distributed, parallel, or sublinear algorithms, as well as the field of descriptive combinatorics, and briefly discusses how techniques from local algorithms can help there, discussing a few cherry-picked examples per section. 
These sections are heavily influenced by the author's particular interests; a different author would choose different applications. 




\subsection{Distributed Computing (CONGEST)}
\label{subsec:congest}

The CONGEST model \cite{peleg2000book} is a model of distributed computing that is extremely tightly connected to local algorithms\footnote{Local algorithms are often referred to as ``distributed algorithms in the LOCAL model of distributed computing''.} Starting from the definition of a local algorithm as a message-passing protocol between nodes in a graph, the CONGEST model adds an additional requirement that each message is supposed to be \emph{small}. In particular, it should fit into $O(\log n)$ bits. For example,  a node can send its unique identifier in one round. 

Many local algorithms in the literature are in fact stated as CONGEST algorithms since a number of local algorithms are readily implemented in the CONGEST model. However, not all of them. 
For example, the problem of counting for each node the number of triangles that contain it is a trivial problem for local algorithms but a very complex problem in the CONGEST model, where it requires polynomially many rounds \cite{chang2019triangle_detection,chang2019triangle_enumeration}. 
The general derandomization technique of \citet*{ghaffari_harris_kuhn2018derandomizing} from \cref{thm:derandomization} is also not directly applicable to the CONGEST model, since computing conditional expectations in a neighborhood of a node requires collecting a lot of information from that neighborhood. However, the general approach is still helpful, and it was used to derandomize concrete problems like maximal independent set \cite{censor2017deterministic_mis_congest} or $\Delta+1$-coloring \cite{bamberger2020efficient}. 

It is an interesting question for which classes of graphs the definitions of local and CONGEST algorithms coincide for local problems. It is known that this is the case for local problems on trees \cite{balliu2021congest_vs_local_for_lcl}. 



\subsection{Local Computation Algorithms and the Volume Model}
\label{subsec:lca}

The model of local computation algorithms \cite{rubinfeld2011LCA_definition_paper,alon2012LCA_definition_paper}, often denoted as LCA, is a model of sublinear algorithms closely related to the local model. We will first explain its variant, known as the \emph{Volume} model by \citet*{RosenbaumSuomela2020volume_model}. This model is very similar to local algorithms, but we measure the \emph{volume} instead of the \emph{radius}. A volume algorithm run at a node $u$ starts with a set $S_0 = \{u\}$ and in the $i$-th step, the algorithm can pick a node $u_i \in S_i$ and an arbitrary index $1 \le j \le d$ where $d$ is the degree of $u_i$. Then, it learns the identifier of the $j$-th neighbor $v$ of $u_i$. We then set $S_{i+1} = S_i \cup \{v\}$. The volume complexity of the algorithm is the number of queries it makes until it decides to finish and output the label for $u$. 

In particular, we note that for graphs of maximum degree $\Delta$, any $t(n)$-round local algorithm can be turned into a volume algorithm of complexity $(2\Delta)^{t(n)}$ as the volume algorithm at $u$ can simply query all the nodes in the $t(n)$-hop neighborhood of $u$ \cite{parnas_ron2007approximating}. 

A local computational algorithm may have several definitions. The stateless local computation algorithm is defined as follows: we start with the definition of the volume model and add the additional requirement that the nodes start with unique identifiers from the range $[n]$, even for randomized algorithms. In each query, the algorithm can additionally ask for the node with the identifier $j$ for any $1 \le j \le n$. The node $u_i$ with identifier $j$ is then added to the set of already seen nodes $S_i$. That is, the algorithm can use additional global queries that are, however, ``blind''. 

It is unclear whether global queries can help with solutions to local problems. We know that they do not help for very fast randomized local computation algorithms \cite{goos2016non}. 
\begin{problem}
    \label{prob:far_queries}
    Is there a local problem that distinguishes (either randomized or deterministic) Local computation algorithms and Volume models with respect to a natural class of graphs? 
\end{problem}

As an example problem considered in the LCA model, there is a long line of work on local computation algorithms for maximal independent set and related problems \cite{
nguyen_onak2008matching,yoshida2009improved,onak2012near,behnezhad2022time,rubinfeld2011LCA_definition_paper,alon2012LCA_definition_paper,levi2015lca_mis,ghaffari2016MIS,reingold2016new,ghaffari2019sparsifying,ghaffari2022mis_in_lca} culminating with the randomized algorithm of \citet*{ghaffari2022mis_in_lca} with volume complexity $\poly(\Delta \log n)$. 

For constant degree graphs, we can try to prove classification results similar to the classification of local problem complexities from \cref{thm:classification_basic}. We have some understanding of the volume complexities of the three important local complexity classes from \cref{thm:classification_basic}:


\begin{restatable}[\cite{parnas_ron2007approximating,even_medina_ron2014coloring_lca,RosenbaumSuomela2020volume_model,grunau_rozhon_brandt2022speedups_below_log_star,brandt_grunau_rozhon2021LLL_in_LCA}]{theorem}{volumeclassify}
\label{thm:classification_volume}
In the class of bounded-degree graphs, each local problem satisfies the following:
\begin{enumerate}
    \item If its round complexity is $O(1)$ (order-invariant regime), then its randomized/deterministic volume complexity is $O(1)$. 
    \item If its randomized/deterministic round complexity is $\omega(1)$ and $O(\log^* n)$ (symmetry-breaking regime), its randomized/deterministic volume complexity is $\Theta(\log^* n)$.  
    \item If its randomized local complexity is $\omega(\log^* n)$ but $o(\log n)$ (Lovász local lemma regime), its randomized volume complexity is between $\Omega(\sqrt{\log n})$ and $O(\log n)$. 
\end{enumerate}
\end{restatable}

We note that as far as we know, all local problems in the local-lemma regime have randomized volume complexity of $\Theta(\log n)$. In particular, sinkless orientation has this randomized volume complexity \cite{brandt_grunau_rozhon2021LLL_in_LCA}. 

\begin{problem}[See Conjecture 1.3 in \cite{brandt_grunau_rozhon2021LLL_in_LCA}]
\label{prob:lca_speedup}
    Is it true on bounded degree graphs that local problems from the local-lemma regime have randomized volume complexity $\Theta(\log n)$?  
\end{problem}

We also note that in the model of deterministic volume complexity where the identifiers come from exponential, instead of polynomial range, we can prove that only the volume complexities $O(1), \Theta(\log^* n), \Theta(n)$ are possible \cite{brandt_grunau_rozhon2021LLL_in_LCA}. The same could be the case also for the standard, polynomial-range, identifiers. 
\begin{problem}
    \label{prob:lca_exponential_vs_polynomial_ids}
    Are there any local problems such that their deterministic volume complexity (with polynomial-sized identifiers) on bounded degree graphs is $\omega(\log^* n)$ but $o(n)$?  
\end{problem}
Equivalently, we could have asked whether there is a local problem whose deterministic volume complexity is different for polynomial-sized and exponential-sized identifiers. 


\subsection{PRAM}
\label{subsec:pram}

PRAM is a classical model of parallel algorithms studied extensively in the past 40 years \cite{jaja1992parallel_algorithms_book}. It simplifies the complexity of practical parallel computing by assuming a simple model of a machine with multiple processors sharing a common memory. 
There are two complexity measures: \emph{work}, i.e., the total number of instructions made by all processors together throughout the execution, and \emph{depth}, i.e., the number of rounds necessary to finish the computation in the case that the machine is equipped with as many processors as the algorithm requires. 

Many techniques developed for distributed and local algorithms are useful in the design of parallel algorithms. 
As an example application for a local problem, \citet*{ghaffari_hauepler2021MIS} develop a randomized algorithm for maximal independent set on the so-called EREW variant of the PRAM model with the optimal depth $O(\log n)$ and $O(m \log^2 n)$ work, building on top of the local maximal independent set algorithm of \citet*{ghaffari2016MIS}. 

As another example application for a non-local problem, we note that one can generalize deterministic local algorithms for network decompositions discussed in \cref{sec:1network_decomposition} to get various clustering results for weighted undirected graphs like the following one. 

\begin{theorem}[Deterministic low-diameter clustering \cite{elkin_haeupler_rozhon_grunau2022Clusterings_LSST}]
    \label{thm:det_sparse_neighborhood_cover}
    Let $G$ be a graph and $w$ its nonnegative weights. Then, there is a parallel algorithm with $\tilde{O}(m+n)$ work and $\tilde{O}(1)$ depth such that, given a parameter $R > 0$, it splits vertices of $G$ into clusters such that 
    \begin{enumerate}
        \item Each cluster has weighted diameter $\tilde{O}(R)$,
        \item The total weight of edges that cross between different clusters is at most $\frac{1}{R} \sum_{e \in E(G)} w(e)$. 
    \end{enumerate}
\end{theorem}
Variants of this clustering result can be used for the design of deterministic parallel algorithms for the approximate shortest path problem or various metric embedding problems \cite{rozhon_grunau_haeupler_zuzic_li2022deterministic_sssp,elkin_haeupler_rozhon_grunau2022Clusterings_LSST}.

\subsection{Massively Parallel Computing (MPC)}
\label{subsec:mpc}

The final parallel model we discuss is the model of Massively parallel computing (MPC) \cite{karloff2010mpc}, the theoretical model behind the popular programming framework MapReduce and its variants. While the previous PRAM model focuses on the total work done by processors and ignores the cost of communication, the MPC model ignores the work done by processors and focuses on the communication cost. 

Concretely, let us explain the so-called low-memory regime of MPC: At the beginning of any graph algorithm, the edges of the input graph are split into many machines, each capable of storing $O(n^\eps)$ bits in its memory. 
In one round, each machine can perform an arbitrary computation on its edges, and then it sends arbitrary information to any of its neighbors (the total information sent and received per machine is $O(n^\eps)$). We measure the number of rounds until the solution is computed. 

\citet*{chang2019coloring_in_MPC} constructed a randomized massively parallel algorithm for $\Delta+1$-coloring that works in $O(\log\log\log n)$ rounds by simulating the fastest randomized local algorithm for coloring of round complexity $\poly\log\log n$ (see \cref{chap:3_concrete_results}) with exponential speedup. 
This result is complemented by the work of \citet*{ghaffari2019conditional_hardness_in_mpc} that showed that this result cannot be improved, conditioned on a certain conjecture. 
Notice how understanding where the three logarithms are coming from requires an extensive understanding of coloring with local algorithms: The way we arrive at this complexity is that, first, starting with the trivial constant-round sequential local algorithm for coloring, we turn it into a deterministic distributed $\poly\log n$-round local algorithm using the general translation of \cref{thm:sequential_vs_distributed_complexity}. Next, this algorithm can be turned into an exponentially faster randomized algorithm using the shattering technique \cite{chang_li_pettie2018optimal_coloring}. Finally, this randomized algorithm is simulated, again with exponential speedup, in the massively parallel model. 

There is also a line of work aiming to extend the classification of local problems to the massively parallel algorithms model, though the complexity landscape there is currently much less understood \cite{balliu2022exponential_speedup_mpc}.

%% file: chapter4/2applications_to_descriptive.tex
\subsection{Descriptive Combinatorics}
\label{sec:4applications_descriptive_combinatorics}

Descriptive combinatorics  \cite{laczkovich1990circle_squaring, marks_unger2017borel_circle_squaring,grabowski_mathe_pikhurko2017circle_squaring,doughertyforeman,marks2016baire,gaboriau,kechris_solecki_todorcevic1999descriptive_combinatorics,marks2016determinacy,millerreducibility,conley_grebik_pikhurko2020divisibility_of_spheres,csoka_grabowski_mathe_pikhurko_tyros2022borel_lll,bernshteyn2020LLL} is an area at the intersection of combinatorics, measure theory, and set theory (see e.g. the surveys \cite{kechris_marks2016descriptive_comb_survey,pikhurko2021descriptive_comb_survey,bernshteyn2022descriptive_vs_distributed_survey}). It studies graphs that arise when one manipulates uncountably-infinitely large mathematical objects equipped with measure. The main object of interest is a \emph{measurable graph}. Instead of a formal definition, let us discuss a particular example. 

\paragraph{An example of a measurable graph}
Consider a circle, i.e., a set of points of distance $1$ from the origin in the plane $\R^2$. A rotation of this circle by $1$ radian counterclockwise induces a graph $G_0$: We can draw an oriented edge from each point $v$ on the circle to the point $v'$ such that the rotation maps $v$ to $v'$. 
This graph has uncountably many connected components, each of which is a doubly-infinite oriented path: this is because $1$ is an irrational multiple of $2\pi$ and thus if we start ``jumping'' from $v$ around the cycle with jumps of length $1$, we never return to $v$. 
Moreover, the standard Lebesgue measure on the circle is telling us which subsets of vertices of $G_0$ we are ``allowed to talk about''. The graph $G_0$, together with the measure on top of it, is an example of a measurable graph. 

Given a local problem $\Pi$, we can now ask whether it can be solved on a particular measurable graph. For example, consider the following problem: Color each vertex of $G_0$ with one of $c$  colors such that the coloring is proper and each set of vertices of the same color is measurable. It is a folklore fact that this type of coloring is not possible with $c = 2$ colors, but it is possible with $c = 3$ colors (this is a special case of the following \cref{thm:cycle}).

\paragraph{Formal connection between the two worlds}
There is a close connection between local algorithms and descriptive combinatorics.  Although this was to some extent understood earlier (see \citet*{elek2010borel}, \citet*[Chapter 23.3]{lovasz2012large}), it was only a breakthrough paper of \citet*{bernshteyn2020LLL} that first realized the full power of the connection. Bernshteyn was able to prove general theorems of the type ``if a local problem $\Pi$ can be solved with a local algorithm of round complexity $O(\log^* n)$, then any measurable graph admits a Borel solution to $\Pi$''.

As we did not define the terms ``measurable graph'' or ``Borel solution'' precisely, we will next only state perhaps the simplest concrete example manifesting the connection, and we point an interested reader to the recent survey of \citet{bernshteyn2022descriptive_vs_distributed_survey} that covers this very exciting and recent topic in depth. 






\begin{theorem}[\cite{grebik_rozhon2021paths}]
    \label{thm:cycle}
    Let $\Pi$ be any local problem (without inputs). Then, the following two statements are equivalent: 
    \begin{enumerate}
        \item There exists an $o(n)$-round local algorithm solving $\Pi$ on sufficiently large oriented cycles. 
        \item One can label each vertex of $G_0$ with a label of $\Pi$ such that all vertices satisfy the constraints of $\Pi$, and for each label $\ell$, the set of vertices with that label is Lebesgue-measurable\footnote{One can replace ``Lebesgue measurable'' with ``a union of finitely many intervals'' and the theorem still holds. }. 
    \end{enumerate}
\end{theorem}

By now, we know of several more results of this type, both for general bounded-degree graphs \cite{bernshteyn2020LLL,bernshteyn2021localcont} and for special cases like trees \cite{brandt_chang_grebik_grunau_rozhon_vidnyanszky2021homomorphism_graphs} and grids \cite{gao_jackson_krohne_seward2018continuous_grids,grebik_rozhon2023grids}. 

\paragraph{Transfer of techniques}
Many recent results in the area of descriptive combinatorics are adopting techniques from local algorithms that are, by now, familiar to readers of this text: network decomposition \cite{bernshteyn_yu2023polynomial_growth}, Lovász local lemma \cite{bernshteyn2020LLL,bernshteyn2021localcont,csoka_grabowski_mathe_pikhurko_tyros2022borel_lll,bernshteyn_weilacher2023finite_asymptotic_separation_index,bernshteyn_yu2023polynomial_growth}, derandomizations \cite{bernshteyn2021localcont,grebik_rozhon2023grids}, ID graphs \cite{brandt_chang_grebik_grunau_rozhon_vidnyanszky2021trees,brandt_chang_grebik_grunau_rozhon_vidnyanszky2021homomorphism_graphs}, and others. 

This connection also led to some progress in the area of local algorithms.  
For example, \citet*{bernshteyn2022vizing} used ideas from a measurable Vizing theorem of \citet*{grebik_pikhurko2020measurable_vizing} to construct a fast local algorithm for $(\Delta+1)$-edge coloring problem. 
\citet*{brandt_chang_grebik_grunau_rozhon_vidnyanszky2021trees} adapted a lower bound of \citet*{marks2016determinacy} to give a different (and similarly simple) lower bound for sinkless orientation from \cref{thm:sinkless_orientation_lb}. 

%% file: chapter4/3other_things.tex
\subsection{Other Models}
\label{sec:4other_things}

There are many more models of local/distributed/parallel/sublinear computation that are, in one way or another, related to local algorithms. 
In many of these models, it not only makes sense to try to solve concrete problems but often, one can also hope that the theory of local algorithms, such as the classification of the local problems on bounded-degree graphs, can be extended similarly to, e.g., \cref{thm:classification_volume}. 

\paragraph{Uniform algorithms}
One favorite model of the author is the model of \emph{uniform} local complexity. In this model, a randomized local algorithm does not know the size of the input graph, $n$. It simply looks at larger and larger neighborhoods of a given vertex, until it decides that it has seen enough to compute the output label at the node. This model was independently studied by local algorithms community \cite{Korman_Sereni_Viennot2012Pruning_algorithms_+_uniform_coloring} and community of probabilists \cite{holroyd_schramm_wilson2017,holroyd2017one_dependent_coloring,Spinka2020finitely_dependent_are_finitary} where it is known as \emph{finitary factors of i.i.d}. 

In the $o(\log n)$-local-complexity regime, uniform algorithms can be seen as a more powerful version of classical local algorithms. 
This is because there are certain local problems such that their solution requires a few nodes to see the whole graph, while most of the nodes can output their solution after seeing their $O(1)$-radius neighborhood \cite{grebik_rozhon2023grids}. 
On the other hand, local problems solvable by uniform local algorithms still often admit measurable solutions defined in descriptive combinatorics \cite{grebik_rozhon2021paths,grebik_rozhon2023grids,brandt_chang_grebik_grunau_rozhon_vidnyanszky2021trees}. Thus, uniform local algorithms can be seen as interpolating between the extremely clean picture painted by the classification of $o(\log n)$-round algorithms (\cref{thm:classification_basic}), and the much less well-understood complexity classes coming from descriptive combinatorics.  

The connection between classical local algorithms and uniform ones would be much cleaner if the following problem was resolved (see \cite{brandt_chang_grebik_grunau_rozhon_vidnyanszky2021trees}):

\begin{problem}
\label{prob:uniform_lll}
    Is there a $C$ and a uniform randomized local algorithm on bounded degree graphs for $C$-relaxed Lovász local lemma such that after $\poly\log\log(1/\eps)$ rounds, at least $1-\eps$ fraction of nodes know the solution? 
\end{problem}

\paragraph{Other models}
There are many more models of locality studied in the literature. Let us list a sample of them. 
The list includes the averaged local complexity \cite{feuilloley2017node_avg_complexity_of_col_is_const,barenboim2018avg_complexity,chatterjee2020sleeping_MIS,balliu2023node_averaged_complexity} (closely connected to uniform algorithms), online local model \cite{suomela2023online_local_model,chang2023online_local_model,akbari2024online_local_quantum}, dynamic local model \cite{suomela2023online_local_model,akbari2024online_local_quantum}, local mending model \cite{melnyk2023mending}, supported local model \cite{schmid_suomela2013supported_model,korhonen_paz_rybicki_schmid_suomela2021supported_model}, quantum local model \cite{gavoille2009quantum_physical_locality,gall2018quantum,coiteux2023no_quantum_advantage,akbari2024online_local_quantum}, awake (energy) complexity \cite{chang2018energy_complexity_exp_separation,chang2020energy_complexity_bfs,chatterjee2020sleeping_MIS,ghaffari_portmann2022sleeping_mis,ghaffari_portmann2023energy}, local certification \cite{goeoes11,fraigniaud2012impact,fraigniaud2013deciding_without_ids,Feuilloley2019local_certification_survey}, finitely dependent colorings \cite{holroydliggett2015finitely_dependent_coloring,holroyd2023symmetrization}, or computable combinatorics \cite{qian_weilacher2022descriptive}.

\subsection*{Final Remarks}

This text aims to present the area of local algorithms in a way that highlights what the author sees as its key aspect: Unlike typical subareas of computer science and discrete mathematics that are usually unified by a set of useful techniques and important results, the area of local algorithms extends beyond this norm. 
We have a clean theory of the model of local algorithms that leads to a clear complexity-theoretical picture. We also understand that local algorithms have applications to a number of other models studied by the broader algorithmic community.  
In this sense, local complexity is similar to much more established fields like communication complexity. 
The author believes that the theory has the potential for diverse extensions and applications and encourages the reader to identify and explore the next one. 

